\documentclass{article}

\usepackage[
	persons={Nima, June, Yeganeh, Kiran},
	authors=blocks,
	bibliosources={refs.bib}
]{pomegranate}

\usepackage[utf8]{inputenc}

\DeclareOperator{\diag}
\DeclareOperator{\weight}
\DeclareOperator{\newt}
\DeclareDocumentTextCommand{\SharpP}{}{\Class{\#P}}
\DeclareDocumentMathCommand{\tmix}{}{t_{\text{mix}}}
\DeclareDocumentMathCommand{\dtv}{}{d_{\text{tv}}}
\DeclareDocumentMathCommand{\inflMat}{}{\Psi^{\text{inf}}}
\DeclareDocumentMathCommand{\corMat}{}{\Psi^{\text{cor}}}
\DeclareDocumentMathCommand{\ef}{m m}{{#1\star #2}}
\DeclareDocumentMathCommand{\I}{}{\mathcal{I}}
\DeclareDocumentMathCommand{\B}{}{\mathcal{B}}

\title{Fractionally Log-Concave and Sector-Stable Polynomials: Counting Planar Matchings and More}
\author{Yeganeh Alimohammadi}
\author{Nima Anari}
\author{Kirankumar Shiragur}
\author{Thuy-Duong Vuong}

\affil{Stanford University, \texttt{\{yeganeh,anari,shiragur,tdvuong\}@stanford.edu}}

\begin{document}
	\maketitle
	\begin{abstract}
		We show fully polynomial time randomized approximation schemes (FPRAS) for counting matchings of a given size, or more generally sampling/counting monomer-dimer systems in planar, not-necessarily-bipartite, graphs. While perfect matchings on planar graphs can be counted exactly in polynomial time, counting non-perfect matchings was shown by \textcite{Jer87} to be \SharpP-hard, who also raised the question of whether efficient approximate counting is possible. We answer this affirmatively by showing that the multi-site Glauber dynamics on the set of monomers in a monomer-dimer system always mixes rapidly, and that this dynamics can be implemented efficiently on downward-closed families of graphs where counting perfect matchings is tractable. As further applications of our results, we show how to sample efficiently using multi-site Glauber dynamics from partition-constrained strongly Rayleigh distributions, and nonsymmetric determinantal point processes. 
%https://www.overleaf.com/project/5e7570644f07e60001a3df4c		
		In order to analyze mixing properties of the multi-site Glauber dynamics, we establish two notions for generating polynomials of discrete set-valued distributions: sector-stability and fractional log-concavity. These notions generalize well-studied properties like real-stability and log-concavity, but unlike them robustly degrade under useful transformations applied to the distribution. We relate these notions to pairwise correlations in the underlying distribution and the notion of spectral independence introduced by \textcite{ALO20}, providing a new tool for establishing spectral independence based on geometry of polynomials. As a byproduct of our techniques, we show that polynomials avoiding roots in a sector of the complex plane must satisfy what we call fractional log-concavity; this extends a classic result established by \textcite{Gar59} who showed homogeneous polynomials that have no roots in a half-plane must be log-concave over the positive orthant.
	\end{abstract}
	
	\section{Introduction}
\label{sec:intro}

Let $\mu:\binom{[n]}{k}\to\R_{\geq 0}$ be a density function on the family of subsets of size $k$ out of a ground set of $n$ elements, which defines a probability distribution
\[ \P{S}\propto \mu(S). \]
The goal of this work is to establish properties of $\mu$ that translate into efficient algorithms for sampling from this distribution, and by classical equivalences between approximate counting and sampling \cite{JVV86}, to algorithms for approximately computing the normalizing constant, i.e., the partition function:
\[ \sum_{S}\mu(S). \]

We study a family of local Markov chains that can be used to approximately sample from such a distribution.

\begin{definition}[Down-Up Random Walks]\label{def:local-walk}
	For a density $\mu:\binom{[n]}{k}\to\R_{\geq 0}$, and an integer $\l\leq k$, we define the $k\leftrightarrow\l$ down-up random walk as the sequence of random sets $S_0, S_1,\dots$ generated by the following algorithm:
	\begin{Algorithm*}
		\For{$t=0,1,\dots$}{
			Select $T_t$ uniformly at random from subsets of size $\l$ of $S_t$.\;
			Select $S_{t+1}$ with probability $\propto \mu(S_{t+1})$ from supersets of size $k$ of $T_t$.\;
		}
	\end{Algorithm*}
\end{definition}

This random walk is time-reversible, always has $\mu$ as its stationary distribution, and moreover has positive real eigenvalues \cite[see, e.g.,][]{ALO20}. The special case of $\l=k-1$ has received the most attention, especially in the literature on high-dimensional expanders \cite[see, e.g.,][]{LLP17,KO18,DK17,KM16,AL20, ALO20}. Each step of this random walk can be efficiently implemented as long as $k-\l=O(1)$ and we have oracle access to $\mu$. This is because  the number of supersets of $T_t$ is at most $n^{k-\l}=\poly(n)$, so we can enumerate over all in polynomial time.

Our main result establishes a formal connection between roots of the \emph{generating polynomial} of $\mu$, defined below, and rapid mixing of the $k\leftrightarrow \l$ down-up walks.

\begin{definition}[Generating Polynomial]
To a density $\mu:\binom{[n]}{k}\to\R_{\geq 0}$ we associate a multivariate \emph{generating polynomial} $g_\mu$, which encodes $\mu$ in its coefficients:
\[ g_\mu(z_1,\dots,z_n):=\sum_{S} \mu(S)\prod_{i\in S}z_i. \]
\end{definition}

\begin{figure}
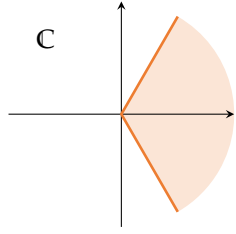

	\Tikz*{
		\fill[LightOrange] (60:1.5) arc(60:-60:1.5) -- (0, 0) -- cycle;
		\draw[-stealth] (-1.5, 0) -- (1.5, 0);
		\draw[-stealth] (0, -1.5) -- (0, 1.5);
		\draw[Orange, line width=1] (60:1.5) -- (0, 0) -- (-60:1.5);
		\node at (-1, 1) {$\C$};
	}
	\caption{A symmetric sector around the positive real axis. Sector-stability of a polynomial means that if all variables are chosen from the interior of the sector, the polynomial does not vanish.}\label{fig:sector}
\end{figure}

Note that $g_\mu$ is a polynomial with nonnegative coefficients, and as such, it has no roots $(z_1,\dots,z_n)\in \R_{>0}^n$. We consider polynomials that not only avoid roots on the positive real axis, but also avoid roots in a neighborhood, that is a sector of the complex plane centered around $\R_{>0}$.

\begin{definition}[Sector-Stability]\label{def:sector-stability}
For an open sector $\Gamma\subseteq \C$ centered around the positive real axis in the complex plane, see \cref{fig:sector}, we call a polynomial $g(z_1,\dots,z_n)$ sector-stable if
\[ z_1,\dots,z_n\in \Gamma\implies g(z_1,\dots,z_n)\neq 0. \]
\end{definition}

Our main result shows that sector-stability where the sector $\Gamma$ has constant aperture, implies rapid mixing of the $k\leftrightarrow \l$ down-up random walk for an appropriately chosen $\l=k-O(1)$.
\begin{theorem}\label{thm:relaxation-time}
	Suppose that the density $\mu:\binom{[n]}{k}\to\R_{\geq 0}$ has a generating polynomial that is sector-stable with respect to a sector $\Gamma$ of aperture $\Omega(1)$. Then for an appropriate value of $\l=k-O(1)$, the $k\leftrightarrow \l$ has relaxation time $k^{O(1)}$.
\end{theorem}
As a reminder, for a time-reversible Markov chain with positive eigenvalues, the relaxation time is the inverse of the spectral gap \cite{LP17}. A corollary of polynomially-bounded relaxation time is that for starting points with not-terribly small probability, the mixing time can be polynomially bounded as well.
\begin{corollary}[{\cite[see, e.g.,][]{LP17}}]
	Suppose $\mu$ has a sector-stable generating polynomial for a sector of constant aperture, and let $\l=k-O(1)$ be the value promised by \cref{thm:relaxation-time}. If the $k\leftrightarrow \l$ down-up random walk is started from $S_0$, then
	\[ \tmix(\epsilon)\leq O\parens*{k^{O(1)}\cdot \log \parens*{\frac{1}{\epsilon\cdot \P_\mu{S_0}}}} \]
	where $\tmix(\epsilon)$ is smallest time $t$ such that $S_t$ is $\epsilon$-close in total variation distance to the distribution defined by $\mu$.
\end{corollary}

As our main application, we obtain efficient algorithms to approximately sample/count (weighted) matchings and matchings of a given size in planar graphs. We discuss this and other applications in \cref{sec:intro-matchings,sec:intro-ndpp,sec:intro-partition}. We then discuss the techniques we use and related work in \cref{sec:spectral-independence,sec:polynomial-properties,sec:glauber}.

\begin{remark}
	\Cref{thm:relaxation-time} can be directly generalized to $\mu$ if the generating polynomial $g_\mu$ has stability w.r.t.\ regions $\Gamma\subseteq \C$ other than a sector. In particular, for $\Gamma=\R_{> 0}\cup D(1, \epsilon)$, where $D(1, \epsilon)=\set{z\in \C\given \abs{z-1}<\epsilon}$ is the disk around $1$, our results still imply $k^{O(1)}$ relaxation time for an appropriate $k\leftrightarrow \l$ down-up walk with $\l=k-O(1)$; see \cref{remark:generalization}. Under this limited assumption, we have spectral independence, but no longer fractional log-concavity (i.e., spectral independence under arbitrary external fields). For clarity of exposition we focus on sector-stability.
\end{remark}

\subsection{Application: Planar Matchings}\label{sec:intro-matchings}

Matchings in graphs have been a rich source of intriguing algorithmic questions. The celebrated blossom algorithm of \textcite{Edm65}, which finds a maximum-sized matching in a general graph, has been partially credited with the creation of the notion of polynomial time algorithm \cite{Koz06}. An entirely different class of algorithms for finding matchings, based on connections to determinants, was introduced by \textcite{Lov79} and developed further by \textcite{KUW86, MVV87}; these determinant-based algorithms have played a central role in the study of parallel algorithms and derandomization \cite[see, e.g.,][]{FGT19}.

Matchings have also played a central role in counting complexity. The problem of \emph{counting} perfect matchings of a given graph was shown by \textcite{Val79} to be complete for the class \SharpP, yielding strong evidence that it cannot be solved in polynomial time. This was the first major result of its kind, demonstrating hardness of counting for a problem whose search version, i.e., the problem of distinguishing zero and nonzero counts, is polynomial-time solvable.

Given the hardness of exact counting \cite{Val79}, the main focus in subsequent work has been on \emph{approximate counting}. Unlike combinatorial optimization problems which often admit nontrivial \emph{approximation factors}, for a wide range of counting problems, the approximation factor achievable in polynomial time can be either made as small as $1+\epsilon$, in fact for inverse-polynomially small $\epsilon$, or it has to be super-polynomially large \cite{SJ89}. Therefore, the gold standard for approximate counting is a fully polynomial time randomized approximation scheme or FPRAS; this is a randomized algorithm whose output is a $(1+\epsilon)$-factor approximation to the count with high probability, running in time $\poly(n, 1/\epsilon)$.

\begin{figure}
	\begin{Columns}<3>
		\Column
		\Tikz*{
			\begin{scope}[every node/.style={
				circle,
				line width=1,
				draw=Black,
				fill=Gray,
				inner sep=3}]
				\node (A) at (-2, 1) {};
				\node (B) at (0, 1) {};
				\node (C) at (2, 1) {};
				\node (D) at (-2, -1) {};
				\node (E) at (0, -1) {};
				\node (F) at (2, -1) {};
			\end{scope}
			\foreach \u in {A, B, C}
				\foreach \v in {D, E, F}
					\draw[line width=1, Black] (\u) -- (\v);
			\foreach \u/\v in {A/D, B/E, C/F}
				\draw[line width=2, Orange] (\u) -- (\v);
		}
		\Column
		\Tikz*{
			\begin{scope}[every node/.style={
				circle,
				line width=1,
				draw=Black,
				fill=Gray,
				inner sep=3}]
				\node (A) at (-30: 0.75) {};
				\node (B) at (90: 0.75) {};
				\node (C) at (-150: 0.75) {};
				\node (D) at (150: 1.5) {};
				\node (E) at (-90: 1.5) {};
				\node (F) at (30: 1.5) {}; 
			\end{scope}
			\foreach \u/\v in {A/B, B/C, C/A, D/B, D/C, E/A, E/C, F/A, F/B}
				\draw[line width=1, Black] (\u) -- (\v);
			\foreach \u/\v in {A/E, B/F, C/D}
				\draw[line width=2, Orange] (\u) -- (\v);
		}
		\Column
		\Tikz*{
			\draw[line width=1, Black] (0, 0) ellipse (2 and 0.5);
			\node at (0, 0) {graph};
			\begin{scope}[every node/.style={
				circle,
				line width=1,
				draw=Black,
				fill=Gray,
				inner sep=3}]
				\node (A) at (-1, 1) {};
				\node (B) at (0, 1) {};
				\node (C) at (1, 1) {};
			\end{scope}
			\foreach \u in {A, B, C}
				\foreach \v in {-0.6, 0, 0.6}
					\draw[line width=1, Black] (\u) -- +(\v, -0.8);
			\draw[decorate, decoration={brace, amplitude=3pt}, line width=1] (-1.1, 1.2) -- (1.1, 1.2) node[pos=0.5, above=5pt] {dummy nodes};
		}
	\end{Columns}
	\begin{Columns}[Top]<3>
		\Column
			\caption{Perfect matchings in bipartite graphs can be approximately counted in polynomial time \cite{JSV04}.}\label{fig:bipartite}
		\Column	
			\caption{Perfect matchings in planar graphs can be exactly counted in polynomial time \cite{Kas61,TF61,Kas67}.}\label{fig:planar}
		\Column
			\caption{Counting matchings of a specified size can be reduced to counting perfect matchings by adding dummy nodes.}\label{fig:reduction}
	\end{Columns}
\end{figure}

In a breakthrough, \textcite{JS89} established an FPRAS for counting matchings of \emph{all sizes} on \emph{unweighted} graphs. It has been a major problem to design an FPRAS for counting matchings of \emph{a given size} or \emph{perfect matchings}. In a celebrated result, \textcite{JSV04} designed an FPRAS for these problems on the important subclass of \emph{bipartite graphs}. Bipartite graphs are an important subclass because of their connection to the permanent of matrices. However, designing an FPRAS to count matchings of a given size on general graphs remains open \cite[see, e.g.,][]{SVW18}.

Besides the class of bipartite graphs, there is another major tractable class for counting \emph{perfect matchings}. Motivated by models in statistical mechanics, \textcite{TF61,Kas61} related the number of perfect matchings in 2-dimensional lattices to a specific determinant, obtaining exact formulae for these counts. Later, \textcite{Kas67} generalized this to all \emph{planar graphs}, obtaining a polynomial time algorithm for \emph{exactly} counting perfect matchings in such graphs. At a high-level, this algorithm finds a suitable signing of the adjacency matrix, a.k.a.\ the Tutte matrix, ensuring its determinant is the square of the number of perfect matchings.

While both bipartite and planar graphs form tractable classes for (approximately/excatly) counting \emph{perfect} matchings, see \cref{fig:bipartite,fig:planar}, there is a major difference between the two when it comes to non-perfect matchings. The problem of counting $k$-matchings, matchings with exactly $k$ edges, is no harder than counting perfect matchings in general. In a general graph on $n$ nodes, one can add $n-2k$ dummy nodes connected to everything else, see \cref{fig:reduction}, and count perfect matchings in the modified graph; the result is $(n-2k)!$ times the number of $k$-matchings. This strategy extends to counting $k$-matchings in bipartite graphs as well. However, in the case of planar graphs, the dummy nodes destroy planarity. This is not just a coincidence. \textcite{Jer87} showed that while perfect matchings can be counted exactly in polynomial time on planar graphs, counting $k$-matchings on such graphs is \SharpP-hard, adding to the mystery of determinant-based counting algorithms. Nevertheless, \textcite{Jer87} raised the possibility of approximately counting $k$-matchings in polynomial time, i.e., designing an FPRAS. As the main application of our results, we resolve this question affirmatively.

\begin{theorem}\label{thm:k-matchings}
	There is a randomized algorithm that receives a planar graph on $n$ nodes and a number $k$, and outputs a $(1+\epsilon)$-approximation to the number of $k$-matchings with high probability, running in time $\poly(n, 1/\epsilon)$.
\end{theorem}

More generally, our results apply to the setting of \emph{weighted} graphs, a.k.a.\ \emph{monomer-dimer systems}. Suppose that a given graph $G=(V, E)$ has edge weights $w:E\to \R_{\geq 0}$ and vertex weights $\lambda:V\to \R_{\geq 0}$. Then define the weight of a matching $M$ as
\[ \weight(M):=\prod_{e\in M} w(e)\cdot \prod_{v\not\sim M} \lambda(v),  \]
where $e$ ranges over dimers, i.e., the matching edges, and $v$ ranges over the monomers, i.e., the vertices not matched in $M$. Normalizing these weights defines a probability distribution over matchings, and approximating the normalizing factor, a.k.a.\ the partition function, is known to be equivalent to approximately sampling from this distribution \cite{JVV86}. It was shown by \textcite{JS89} how to approximately sample/count from monomer-dimer systems in general graphs when edge weights $w(e)$ are polynomially bounded and there are no vertex weights $\lambda(v)$; these assumptions on weights are quite strong, despite their seemingly innocuous appearance. Approximately sampling/counting from the monomer-dimer systems with no restriction on the weights remains a key challenge.

Computing statistics of monomer-dimer systems on 2-dimensional lattices, and more generally planar graphs, was originally studied in statistical physics \cite{Kas61, TF61,Kas67}. However, the determinant-based algorithms found could only solve the case of zero monomer weights: $\forall v: \lambda(v)=0$. Here we remove this restriction, at the expense of approximation.\footnote{We remark that by the results of \cite{Jer87}, approximation appears to be necessary, at least for the counting problem.}
\begin{theorem}\label{thm:monomer-dimer}
	There is an algorithm that receives a planar graph $G=(V, E)$ on $n$ vertices and weights $w:E\to \R_{\geq 0}$ and $\lambda:V\to \R_{\geq 0}$, and outputs a random matching $M$, whose distribution is $\epsilon$-close in total variation distance to the monomer-dimer distribution induced by $w, \lambda$. The running time of this algorithm is $\poly(n, \log(1/\epsilon))$.
\end{theorem}

Our results do \emph{not} rely on planarity strongly. In fact, \cref{thm:k-matchings,thm:monomer-dimer} extend to any downward-closed family of graphs for which perfect matchings can be counted efficiently. Examples that go beyond planar graphs include certain minor-free graphs \cite{EV19}, and small genus graphs \cite{GL99}.

The key insight that enables \cref{thm:k-matchings,thm:monomer-dimer} is that we show local random walks on the set of \emph{monomers}, or terminals of the matching $M$, mix rapidly on all graphs. Monomer-dimer systems and $k$-matchings each induce a distribution on subsets $S$ of vertices of the graph if we only view the unmatched (or dually matched) vertices, i.e., the monomers. On planar graphs, the weight of each such set $S$ can be computed efficiently, up to a global normalizing factor.
\[ \mu(S):=\sum\set{\weight(M)\given M \text{ is a perfect matching on $S^c$}}, \]
where $S^c$ denotes the complement of $S.$

We show how to sample a set $S$ with probability approximately following the above, by running a multi-site Glauber dynamics on $S$ for polynomially many steps. The rapid mixing of this random walk, combined with known equivalences between approximate sampling and approximate counting \cite{MVV87} imply \cref{thm:k-matchings,thm:monomer-dimer}.

We prove \cref{thm:monomer-dimer,thm:k-matchings} by showing sector-stability of the corresponding generating polynomials and then applying \cref{thm:relaxation-time}. We show sector-stability by starting from results of \textcite{HL72} who characterized regions of root-freeness for \emph{unconstrained non-homogeneous} monomer-dimer systems, and applying a set of tools we build that show sector-stability degrades gracefully under a number of operations, like conditioning on cardinality or homogenization.

\begin{lemma}\label{lem:matching-is-sector-stable}
	Suppose that a graph $G=(V, E)$ is given with edge weights $w:E\to \R_{\geq 0}$ and vertex weights $\lambda:V\to \R_{\geq 0}$ which together define a weight on matchings \[\weight(M)=\prod_{e\in M}w(e)\prod_{v\not\sim M}\lambda(v).\] For any $k$, the following polynomial, encoding $k$-matchings, is sector stable for a sector of aperture $\pi/2$.
	\[ g(z_1,\dots,z_n)=\sum_{M\text{ matching of size }k} \weight(M)\prod_{v\not\sim M} z_v. \]
	Additionally the following homogeneous polynomial in $2n$ variables, encoding all matchings, is sector stable for a sector of aperture $\pi/2$.
	\[ g(z_1,\dots,z_n,z'_1,\dots,z'_n)=\sum_{M\text{ matching}}\weight(M)\prod_{v\not\sim M}z_v\prod_{v\sim M}z'_v. \]
\end{lemma}

\begin{remark}
	Techniques developed by \textcite{JS89} allow one to tune the weights in monomer-dimer systems to make the probability mass of $k$-matchings inverse-polynomially large. In turn combining these techniques with rejection sampling, \cref{thm:k-matchings} can be derived from \cref{thm:monomer-dimer}. Nevertheless, our techniques directly solve the sampling problem for $k$-matchings, monomer-dimer systems, and even monomer-dimer systems restricted to $k$-matchings, without the need to resolve to weight-tuning.
\end{remark}

\subsection{Application: Nonsymmetric Determinantal Point Processes}\label{sec:intro-ndpp}

Determinantal point processes (DPP) are elegant probabilistic models used to capture the relationship between items within a subset drawn from a large universe of items. A DPP is formally defined with the help of an $n\times n$ positive semidefinite matrix $L\succeq 0$, where a subset $S\subseteq [n]$ is chosen with probabilities given by minors of $L$:
\[ \P{S}\propto \det(L_{S, S}). \]

 Determinantal point processes (DPP) were first studied in 1975 by Macchi~\cite{Macchi75}, who was motivated by the study of fermion processes in quantum mechanics. Since then, DPPs have been very well-studied and have found applications in many areas such as physics  \cite{CMO19,Sos02}, random matrix theory \cite{Joh05}, combinatorics \cite{BBL09} (random spanning trees \cite{BP93}, non-intersecting paths \cite{Ste90}) and recently in machine learning. Within machine learning, DPPs have been used in several applications such as document summarization \cite{CGGS15, LB12}, recommender systems \cite{GPK16}, and many others~\cite{AFAT14,KT11,KSG08}. Due to broad and practical applications, algorithmic questions occurring in DPP have received lot of attention and efficient algorithms for DPP learning~\cite{AFAT14,Bor09,KT12,LMR15} and sampling \cite{AOR16,RK15,LJS16,HKPV06} have been provided. 

\Textcite{KT11,KT12} studied an extension of DPPs where the samples are conditioned on having a fixed size $k$. These so called $k$-DPPs are formally defined with the help of an $n\times n$ positive semidefinite matrix $L\succeq 0$ and a cardinality parameter $k$, where a subset $S\in \binom{ [n]}{k}$ of size $k$ is chosen with probabilities given by $k\times k$ minors of $L$:
\[ \P{S}= \frac{\det(L_{S, S})}{\sum_{T  \in \binom{ [n]}{k}}\det(L_{T, T})}. \] 
The authors in \cite{KT11,KT12} used $k$-DPPs to attack problems such as the image search task, where the goal is to output a \emph{diverse} set of image results, of desired cardinality, in response to a search query.

Almost all  prior work on DPPs assume the underlying matrix $L$ is symmetric and positive semidefinite (PSD) and the understanding of nonsymmetric DPPs (where $L$ does not have to be symmetric) remains sparse. For nonsymmetric matrices $L$ that are guaranteed to have nonnegative minors, the nonsymmetric DPP can still be defined by
\[ \P{S}\propto\det(L_{S, S}). \]

Nonsymmetric DPPs are important as they allow one to model both repulsive and attractive relationships between items, providing a significantly improved modeling power. For applications of nonsymmetric DPPs see \cite{Gartrell2019LearningND}, where the authors use nonsymmetric DPPs to effectively recover correlation structure within data, particularly for data that contains large disjoint collections of items where the items within the same collection have positive correlation while those across different collections are negatively correlated.  \textcite{BWLGCG18} also studied learning certain subclasses of nonsymmetric DPPs. Due to their enhanced expressivity power and potential new applications, the study of nonsymmetric DPPs has been an active area of research in the past few years.

The question of sampling from nonsymmetric $k$-DPPs is known to be polynomial-time tractable. Indeed, the counting question, that is computing the sum of principal minors can be done exactly, even when restricted to $k\times k$ principal minors. However these naive algorithms are cumbersome to run in practice, as they require \emph{at least} $n\times n$ matrix multiplication time. A similar barrier existed for symmetric DPPs, but Markov-chain-based sampling from $k$-DPPs for symmetric $L$ provided one way to get around this barrier \cite{AOR16,LJS16}, yielding algorithms that run in $O(n\poly(k))$ time.

As an application of our results we provide the first efficient Markov-chain-based algorithm to sample from a wide class of nonsymmetric $k$-DPPs. Our algorithm works for any nonsymmetric matrix $L$ satisfying $L+L^T \succeq 0$. These matrices are the sum of a skew-symmetric matrix and a symmetric PSD matrix; this class of matrices $L$, which are automatically guaranteed to have nonnegative principal minors, defines the main class of nonsymmetric DPPs studied in the literature \cite{Gartrell2019LearningND}.

\begin{theorem} \label{thm:DPPsample}
For any matrix $L \in \R^{n\times n}$ satisfying $L + L^\intercal \succeq 0 $ and cardinality $k\geq 0$, consider the  distribution $\mu: \binom{[n]}{k} \to \R_{\geq 0}$ defined by 
\[\mu(S) \propto \det(L_{S,S}).\]
Then the $k\leftrightarrow (k-4)$ random walk for $\mu$ has relaxation time $\poly(k)$.
\end{theorem}
Note that each step of this random walk can be implemented using $O(n^4)$ computations of $k\times k$ principal minors of $L$. So this results in a mixing time of $O(n^4\poly(k)\cdot \log(1/\P{S_0}))$. To the best of our knowledge, our work is the first to establish that natural Markov chains can be used for the task of sampling from nonsymmetric $k$-DPPs.

Unsurprisingly, we show this result by proving sector-stability of the corresponding generating polynomial.

\begin{lemma}
	For any matrix $L\in \R^{n\times n}$ satisfying $L+L^\intercal\succeq 0$ and number $k$, the following polynomial is sector-stable w.r.t.\ a sector of aperture $\pi/2$.
	\[ g(z_1,\dots,z_n)=\sum_{S\in\binom{[n]}{k}} \det(L_{S,S})\prod_{i\in S}z_i. \]
\end{lemma}

\subsection{Application: Partition-Constrained Strongly Rayleigh Distributions}\label{sec:intro-partition}

Suppose that $\mu:\binom{[n]}{k}\to\R_{\geq 0}$ is a density where $g_\mu$ is stable with respect to a half plane in $\C$, i.e., stable w.r.t.\ the sector $\set{z\in \C\given \Re(z)>0}$. Distributions with this property are called \emph{strongly Rayleigh}, and they have been widely studied in the literature \cite[see, e.g.,][]{BBL09}. Strongly Rayleigh distributions include determinantal point processes, certain classes of matroids, results of the symmetric exchange process, and more \cite[see, e.g.,][]{BBL09}. Motivated by the important problems of computing mixed discriminants, and counting intersections of matroids, several works \cite{AO17,SV17,EDKSV16,KD16} have studied the problem of sampling from such $\mu$ subject to a \emph{partition constraint}. That is, given a partition $T_1\cup T_2\cup \cdots \cup T_s=[n]$, and numbers $c_1,\dots,c_s\in \Z_{\geq 0}$, the question is to sample $S\sim \mu$ conditioned on the constraint
\[ \forall i: \card{S\cap T_i}=c_i. \]
If we allow arbitrarily large $s$, this problem becomes as hard as (approximately) computing the mixed discriminant for which no FPRAS is known. If one defines the same problem for distributions $\mu$ that have a log-concave generating polynomial, then partition-constrained sampling is as hard as sampling from the intersection of \emph{two matroids}; this is again an important open problem, which remains unsolved.

Given the importance of partition-constrained distributions mentioned above, a natural question is, are there assumptions on the partitions that allow for an FPRAS or approximate sampling? \textcite{EDKSV16} obtained such a positive result when the number of partitions $s$ is a constant and \emph{importantly} when $g_\mu$ can be computed exactly (as is the case for determinantal distributions). They relied on polynomial interpolation to achieve this result. However, for many strongly Rayleigh distributions $\mu$, we can only approximately compute $g_\mu$.

As a further application of our results, we show how to sample from partition-constrained $\mu$, as long the number of partitions is $O(1)$; our algorithm only requires having access to an oracle for $\mu$, as opposed to $g_\mu$. We do this by showing that the local random walks on the partition-constrained $\mu$ still mix rapidly, by relying on \cref{thm:relaxation-time} and showing sector-stability for the conditioned distribution.
\begin{lemma} \label{lem:partition-is-stable}
	Suppose that $\mu$ has a sector-stable polynomial with respect to the sector $\set{z\in \C\given \Re(z)>0}$. Then the partition-constrained distribution for $O(1)$-many partitions is sector-stable w.r.t.\ a sector of $\Omega(1)$ aperture.	
\end{lemma}

As a corollary of the ability to approximately compute the partition function for $\mu$ subject on partition constraints, we show how to approximately compute mixed derivatives of real-stable polynomials $g_\mu$, where the number of \emph{distinct} derivatives is $O(1)$. Note that without this restriction of $O(1)$, this problem becomes as hard as computing mixed discriminants.

\begin{corollary}\label{cor:mixDerivative}
Let $g(z_1, \cdots, z_n)$ be a homogeneous real-stable polynomial with nonnegative coefficients.
Suppose we are given oracle access to coefficients of $g$, and we are also given a term with nonzero coefficient. Then there is an FPRAS that can approximately compute mixed derivatives of $g$ along positive directions, as long as the number of \emph{unique} directions is $O(1)$. That is given $v^1, \cdots, v^s, x \in \R_{\geq 0}^n$ with $s=O(1)$ and tuple $(c_1, \cdots, c_s) \in \Z_{\geq 0}^s$, we can efficiently approximate \[\eval{\partial_{v^1}^{c_1} \cdots \partial_{v^s}^{c_s} g}_{z = x}.\]
Here $\partial_v$ is simply the operator $v_1\partial_{z_1}+\dots+v_n\partial_{z_n}$.
\end{corollary}

\subsection{Techniques and Related Work: Pairwise Correlations and Spectral Independence}\label{sec:spectral-independence}

In order to prove \cref{thm:relaxation-time}, we build on a recent line of work leveraging high-dimensional expanders for sampling problems \cite{ALOV19,AL20,ALO20,CLV20a,FGYZ20,CGSV20,CLV20b}. Specifically, we use the framework dubbed spectral independence by \textcite{ALO20}. In this framework, one views a target distribution $\mu$ as a weighted hypergraph or simplicial complex. Establishing a certain notion of high-dimensional expansion would then imply fast-mixing of natural random walks that converge to $\mu$ \cite{DK17,KM16,LLP17,KO18,AL20}. Reinterpreting the notion of high-dimensional expansion needed for rapid mixing, \textcite{ALO20} showed how properties of \emph{pairwise correlations} in the distribution $\mu$, and certain distributions derived from $\mu$, can imply rapid mixing of natural local random walks, see \cref{def:local-walk}.

The spectral independence framework can be applied to the problem of sampling from a distribution on size $k$ subsets of a ground set of $n$ elements, given up to a global normalizing factor by a function $\mu$:
 \[\mu:\binom{[n]}{k}\to\R_{\geq 0}.\]
In many cases the domain of $\mu$ can be adapted to be of the form $\binom{[n]}{k}$ \cite[see, e.g.,][]{ALO20}. For concreteness, let us look at the distribution of monomers in a monomer-dimer system on the graph $G=(V, E)$. Not all monomer sets have the same size, but we can view each set $S\subseteq V$ as a subset of size $\card{V}$ chosen from $V\times \set{0, 1}$:
\[ S\mapsto \set{(v, 0)\given v\notin S}\cup \set{(v, 1)\given v\in S}. \]
This gives us a distribution $\mu:\binom{V\times\set{0,1}}{\card{V}}\to \R_{\geq 0}$. Note that in the case of $k$-matchings, the monomer set is already of a fixed size, and there is no need for this transformation.

\Textcite{ALO20}, based on earlier work of \textcite{AL20}, showed that rapid mixing of natural local random walks converging to $\mu$ can be established as long as pairwise correlations of $\mu$ (and certain distributions derived from $\mu$) are \emph{spectrally bounded}. More precisely, consider the correlation matrix defined below.
\begin{definition}[Correlation Matrix]\label{def:correlation-matrix}
	For a distribution $\mu$ over subsets $S$ of a ground set $[n]$, define the correlation matrix $\Psi\in \R^{n\times n}$ as the matrix having entries
	\[ \Psi_{i, j}:=\P_{S\sim \mu}{j\in S\given i\in S}-\P_{S\sim \mu}{j\in S}. \]
\end{definition}
The entries of the matrix $\Psi$ measure pairwise correlations or in other words deviations from pairwise independence.\footnote{We remark that in some works using the spectral independence framework, the matrix $\Psi$ is defined slightly differently, with entries of the form $\P_{S\sim \mu}{j\in S\given i\in S}-\P_{S\sim \mu}{j\in S\given i\notin S}$, but these matrices are directly related, and we believe it is more natural to consider the definition presented here.} The key behind the spectral independence framework is to show that the maximum eigenvalue of $\Psi$ is $O(1)$. Note that $\Psi$ is always similar to a symmetric matrix and therefore has real eigenvalues \cite{ALO20}. More precisely, one needs to show this not just for the distribution $\mu$, but also conditioned versions of it. We remark that in earlier work, a variant of the correlation matrix has appeared where the entries are instead given by $\P{j\in S\given i\in S}-\P{j\in S\given i\notin S}$, but these two variants are intimately connected and for homogeneous distributions one can go from eigenvalue bounds of one to the other.
\begin{definition}[Conditioned Distribution]\label{def:conditioned}
	For a distribution $\mu$ defined over subsets of a ground set $[n]$ and $T\subseteq [n]$, define $\mu_T$ to be the distribution of $S\sim \mu$ conditioned on the event $T\subseteq S$.
\end{definition}
One has to show that the correlation matrix has bounded eigenvalues for every $T$ where $\mu_T$ is well-defined. The main challenge in all applications of this framework is bounding these eigenvalues. Roughly speaking, prior work has managed to use three categories of techniques to establish eigenvalue bounds, discussed below:
\paragraph{Trickle-Down.} \textcite{Opp18} showed that an eigenvalue bound on $\Psi$ for $\mu_{\set{1}},\mu_{\set{2}},\dots ,\mu_{\set{n}}$ also implies an eigenvalue bound for $\Psi$ for the distribution $\mu$, under some mild additional conditions. This enables an inductive approach to bounding the eigenvalues of $\Psi$, starting from $\mu_T$ for large sets $T$ (i.e., of size $k-2$). The main challenge here has been that in most cases, the eigenvalue bound \emph{deteriorates}, and the induction cannot be completed. A notable exception to this deterioration of the bounds are distributions related to matroids \cite{ALOV19}, but as was observed by \textcite{AL20}, for almost any distribution beyond matroids and matroid-related ones, one has to employ additional tricks to make this induction useful for sampling.
\paragraph{Negative Correlation.} Some distributions have negative entries in $\Psi$, everywhere except on the diagonal; this property is known as negative correlation \cite[see, e.g.,][]{BBL09}. Most notably, the uniform distribution on spanning trees, balanced matroids, and determinantal point processes, all have negative correlation \cite{FM92,BBL09}. When negative correlations exist, the $\ell_1$ norm of rows of $\Psi$ and consequently its maximum eigenvalue can be bounded by $O(1)$ \cite{ALO20}. For non-homogeneous distributions that satisfy negative correlation, related statements hold, as was shown recently by \textcite{ES20}.

\paragraph{Correlation Decay.} When $\mu$ is a distribution defined on an underlying graph, e.g., spin systems which are distributions on random assignments $\sigma:V\to [q]$ of $q$ spins to vertices of a graph, one can define a class of properties under the umbrella term ``correlation decay''. Informally, these properties imply that for distant vertices $u, v$, the values of $\sigma(u), \sigma(v)$ are almost independent of each other. Naturally this is very useful for bounding the entries and consequently the eigenvalues of the matrix $\Psi$. While correlation decay properties were known to yield efficient sampling/counting algorithms, when combined with the spectral independence framework, they resulted in algorithms with truly polynomial running times (compared to prior results which often needed extra assumptions such as boundedness of the degrees in the graph) for several problems like the hardcore model \cite{ALO20}, two-spin systems \cite{CLV20a}, and random colorings \cite{CGSV20, FGT19}.

\begin{figure}
	\begin{Columns}<3>
		\Column
		\Tikz*{
			\begin{scope}[every node/.style={
			circle,
			line width=1,
			draw=Black,
			fill=Gray,
			inner sep=3}]
				\node (A) at (0, 2) {};
				\node (B) at (0, -2) {};
			\end{scope}
			\draw[line width=1, Black] (A) -- (B);
			\node (C) at (1.2, 0) {correlated};
			\draw[-stealth, dashed, line width=1] (C) to[bend right] (A);
			\draw[-stealth, dashed, line width=1] (C) to[bend left] (B);
		}
		\Column
		\Tikz*{
			\node (C) at (0, 0) {$\vdots$};
			\begin{scope}[every node/.style={
			circle,
			line width=1,
			draw=Black,
			fill=Gray,
			inner sep=3}]
				\node[label={[left]$1$}] (A) at (0, 2) {};	
				\node[label={[left]$0$}] (B) at (0, 1) {};
				\node[label={[left]$0$}] (D) at (0, -1) {};
				\node[label={[left]$1$}] (E) at (0, -2) {};
			\end{scope}
			\foreach \u/\v in {A/B, B/C, C/D, D/E}
				\draw[line width=1, Black] (\u) -- (\v);
			\node (T) at (1.2, 0) {correlated};
			\draw[-stealth, dashed, line width=1] (T) to[bend right] (A);
			\draw[-stealth, dashed, line width=1] (T) to[bend left] (E);
		}
		\Column
		\Tikz*{
			\begin{scope}[every node/.style={
			circle,
			line width=1,
			draw=Black,
			fill=Gray,
			inner sep=3}]
				\node[fill=Navy] (A) at (-0.5, 0) {};
				\node[fill=Orange] (B) at (1, 2) {};
				\node[fill=Orange] (C) at (1, 1) {};
				\node[fill=Orange] (D) at (1, 0) {};
				\node (E) at (1, -1) {};
				\node (F) at (1, -2) {};
			\end{scope}
			\foreach \u in {B, C, D}
				\draw[stealth-stealth, dashed, line width=1] (A) -- (\u);
			\draw[decorate, decoration={brace, mirror, amplitude=3pt}, line width=1] (1.2, -0.1) -- (1.2, 2.1) node[pos=0.5, rotate=-90, above=5pt] {not many};
		}
	\end{Columns}
	\begin{Columns}[Top]<3>
		\Column
			\caption{The two vertices are either both monomers or neither are. Therefore they are \emph{positively} correlated.}\label{fig:positive-correlation}
		\Column
			\caption{Only two matchings, one with odd edges and one with even edges, appear in the monomer-dimer system. The endpoints have long-range correlation.}\label{fig:long-range-correlation}
		\Column
			\caption{Informally, the number of vertices strongly correlated with any given vertex is bounded.}\label{fig:bounded-correlation}
	\end{Columns}
\end{figure}

Unfortunately, in the case of the monomer distribution in monomer-dimer systems, none of these methods appear to work.\footnote{We remark that for the special case of unweighted monomer-dimer systems, a form of correlation decay does exist \cite{BGKNT07}.} As demonstrated in \cref{fig:positive-correlation,fig:long-range-correlation}, we can have both positive and long-range correlations. Nevertheless, we show that the correlation matrix is still bounded, see \cref{fig:bounded-correlation}.

\begin{theorem}\label{thm:sector-implies-bounded}
	Suppose that $\mu:\binom{[n]}{k}\to\R_{\geq 0}$ is a density whose generating polynomial is sector-stable w.r.t.\ a sector of aperture $\Omega(1)$. Then the $\l_1$ norm of any row in the correlation matrix $\Psi$ is bounded by $O(1)$.
	\[ \forall i:\;\sum_{j} \abs*{\P_{S\sim \mu}{j\in S\given i\in S}-\P_{S\sim \mu}{j\in S}}\leq O(1). \]
\end{theorem}

Note that a bound on the $\l_1$ norm of rows, is also a bound on the maximum eigenvalue \cite[see, e.g.,][]{ALO20}. Combining this with sector-stability of various distributions, e.g., the monomer distribution, results in specific bounds on the correlation matrix.

\begin{corollary}\label{cor:monomer-bounded-correlations}
	Let $\mu$ be the distribution of monomers in uniformly random $k$-matchings or more generally monomer-dimer systems with arbitrary weights (possibly restricted to $k$-matchings). Then the $\ell_1$ norm of rows of the correlation matrix $\Psi$ are bounded by a universal constant:
	\[ \forall i:\;\sum_{j} \abs*{\P_{S\sim \mu}{j\in S\given i\in S}-\P_{S\sim \mu}{j\in S}}\leq O(1). \]
\end{corollary}

Our main technical contribution is introducing a new technique for establishing spectral independence based on the roots of the partition function in the complex plane.

\subsection{Techniques and Related Work: Sector-Stability and Fractional Log-Concavity}
\label{sec:polynomial-properties}

The study of roots of polynomials associated with distributions has a very long history, most notably in statistical physics. In statistical physics, having roots near the positive real axis is recognized as an indicator of phase transition. This is because roots indicate singularity of $\log g_\mu$, and many physical observables are related to $\log g_\mu$ and its derivatives, which can rapidly change near singularities \cite[see, e.g.,][]{YL52}. For monomer-dimer systems, \textcite{HL72} established a crucial property for roots of the polynomial defined below:
\[ \sum_{M\text{ matching}}\weight(M)\prod_{v\not\sim M}z_v. \]
Here for each matching $M$, we multiply its weight by the variables $z_v$, for $v$ ranging over the monomers. \textcite{HL72} formally showed that if we plug in $z_1,\dots, z_n\in \C$ such that $\Re(z_1),\dots,\Re(z_n)>0$, then the above expression will not result in zero. This is the crucial property that \cref{lem:matching-is-sector-stable} and consequently \cref{thm:k-matchings,thm:monomer-dimer} rely on. This property is also known as Hurwitz-stability \cite[see, e.g.,][]{BB09}.

Note that the polynomial defined by \textcite{HL72} is not homogeneous, i.e., it does not correspond to a distribution on $\binom{[n]}{k}$. Unfortunately, homogenization does not preserve Hurwitz-stability; similarly we do not get Hurwitz-stability if we only include matchings $M$ of a particular size. We establish the weaker, but more robust, notion of sector-stability for these polynomials. Instead, we show that monomer distributions, when homogenized or conditioned on size, cannot have roots in a wide enough sector in \cref{lem:matching-is-sector-stable}.

A special case of sector-stability, when the sector is the entire right-half-plane, is equivalent to Hurwitz-stability. For \emph{homogeneous} polynomials, Hurwitz-stability is the same as another widely studied property called real stability, or more generally, the so-called half-plane property \cite{BBL09}. Under this special notion of sector-stability, the distribution $\mu$ is known to exhibit negative correlations \cite{BBL09}, and rapid mixing of local random walks for $\mu$ had already been established \cite{FM92,AOR16}. Outside of this special case, negative correlation no longer holds. But we show that correlations are still bounded in \cref{thm:sector-implies-bounded}.

As mentioned before, real-stability, a special case of sector-stability for homogeneous polynomials, is a well-studied property of the generating polynomial $g_\mu$ that already implied efficient sampling or counting algorithms for $\mu$ \cite[see, e.g.,][]{AOR16}. However, recent works have shone light on a generalization of real-stability, that does \emph{not} involve root locations. \textcite{ALOV19} established that if $\log g_\mu(z_1,\dots,z_n)$ is concave, viewed as a function over $\R_{\geq 0}^n$, then $k\leftrightarrow (k-1)$ down-up random walks for sampling from $\mu$ would rapidly mix. This class of \emph{log-concave polynomials} have been instrumental in resolving several long-standing questions about matroids \cite{ALOV18,ALOV19,BH19}.

In prior work, \textcite{MS19} established central limit theorems under \emph{univariate} sector-stability of the generating polynomial associated with distributions supported on $\Z$, with extra assumptions on the variance of these distributions. While these results are in the same spirit as bounds we get on correlations, we do not know of a formal connection. A key difference in this work is that we deal with \emph{high-dimensional} distributions and \emph{multivariate} polynomials, and make no assumptions beyond sector-stability; in contrast, to get central limit type theorems, one has to at least make the assumption that the variance grows to infinity.

In other related work, \textcite{Wag09} established generalizations of the result of \textcite{HL72} on Hurwitz-stability of monomer-dimer distributions, showing that certain polynomials enumerating spanning subgraphs with degree constraints are sector-stable. We leave the question of deriving algorithmic applications of this sector-stability to future work.

Log-concave polynomials are a proper superset of real-stable polynomials, at least in the homogeneous case. This was first shown by \cite{Gar59}, and this important result has been instrumental in the development of hyperbolic programming \cite{Gul97}. A natural question that arises is, whether there is an analogous generalization of log-concavity, that is a superset of sector-stable polynomials.

We define a natural property, that we call \emph{fractional log-concavity}. We show that in a ``local sense'', it is actually equivalent to spectral independence of the distribution $\mu$, and then show that sector-stability implies fractional log-concavity, establishing an extension of the result of \textcite{Gar59}.

\begin{definition}[Fractional Log-Concavity]\label{def:fractional-log-concavity}
We call the polynomial $g_\mu(z_1,\dots,z_n)$ fractionally log-concave with parameter $\alpha\in [0, 1]$, if $\log g_\mu(z_1^\alpha,\dots,z_n^\alpha)$ is concave, viewed as a function over $\R_{\geq 0}^n$.
\end{definition}

Note that for $\alpha=1$, this is the same as log-concavity. We show the following \emph{local} equivalence between spectral independence and fractional log-concavity.
\begin{proposition}
	Suppose that $\mu:\binom{[n]}{k}\to\R_{\geq 0}$ is a distribution, and define the $n\times n$ correlation matrix $\Psi$ as 
	\[\Psi_{i, j}:=\P_{S\sim	\mu}{j\in S\given i\in S}-\P_{S\sim \mu}{j\in S}.\]
	Then the maximum eigenvalue of $\Psi$ is bounded by $O(1)$ if and only if the polynomial $g_\mu$ is fractionally log-concave around the point $z=(1,\dots,1)$ for a parameter $\alpha>\Omega(1)$.
\end{proposition}

This combined with \cref{thm:sector-implies-bounded} shows that sector-stable polynomials are fractionally log-concave around the point $(1,\dots,1)$. However, sector-stability is preserved under the change of variables $z_i\mapsto \lambda_i z_i$, where $\lambda_1,\dots,\lambda_n$ are positive reals. This is because sectors in the complex plane are preserved under such scalings. This allows us to map any point in $\R_{\geq 0}^n$ to the special point $(1,\dots,1)$. Using this we establish an extension of the result of \textcite{Gar59}.
\begin{theorem}
	Suppose that $g_\mu$ is sector-stable for a sector of aperture $\Omega(1)$. Then $g_\mu$ is fractionally log-concave for a parameter $\alpha\geq\Omega(1)$.
\end{theorem}
As a corollary of this result we prove bounds similar to those obtained by \textcite{AOV18} relating the entropy of fractionally log-concave, and consequently sector-stable, distributions with the sum of their marginal entropies. See \cref{sec:log-concave}.

While fractional log-concavity around the point $(1,\dots,1)$ is equivalent to a bound on the eigenvalues of the correlation matrix $\Psi$, it does \emph{not} imply a bound for the conditioned distributions $\mu_T$. However fractional log-concavity at \emph{all points} in $\R_{\geq 0}^n$ does. This is because the polynomial for conditional distributions $\mu_T$ can be obtained as the following limit:
\[ g_{\mu_T}\propto \lim_{\lambda\to \infty} g_\mu\parens*{\overbrace{\lambda z_1,\lambda z_2,}^{\text{elements in }T}\dots,z_n}/\lambda^{\card{T}}. \]
Scaling the variables or the polynomial, and taking limits all preserve fractional log-concavity.

\begin{corollary}
	If $\mu:\binom{[n]}{k}\to\R_{\geq 0}$ has a fractionally log-concave generating polynomial with parameter $\alpha=\Omega(1)$, or a sector-stable polynomial with a sector of aperture $\Omega(1)$, then for all conditioned distributions $\mu_T$, the correlation matrix has maximum eigenvalue $O(1)$.
\end{corollary}

This work establishes a number of examples of fractionally log-concave polynomials, but all of our examples are also sector-stable. We leave the question of finding other examples of fractionally log-concave polynomials that go beyond sector-stability to future work. However, we make the following concrete conjecture, in line with a conjecture of \textcite{MV89} on the expansion of $0/1$ polytopes.

\begin{conjecture}
	Suppose that $\mu$ is the uniform distribution on a subset of the hypercube $F\subseteq \set{0,1}^n$, such that the convex hull $\conv(F)$ has edges of bounded length $O(1)$. Then we conjecture that the polynomial
	\[ \sum_{S\in F}\mu(S)\prod_{i\in S}z_i \]
	is fractionally log-concave for a parameter $\alpha>\Omega(1)$.
\end{conjecture}

Matroids are a special case of this conjecture, and their log-concavity has already been established \cite{AOV18}. However this conjecture is widely more general, encompassing combinatorial objects such as delta-matroids, Coxeter matroids, and more \cite{BGW03}.

\subsection{Techniques and Related Work: Multi-Site Glauber Dynamics}\label{sec:glauber}
All of our sampling algorithms are obtained as instantiations of the $k\leftrightarrow \l$ down-up random walk for some $\l=k-O(1)$ applied to an appropriate formulation of the target distribution $\mu$, see \cref{def:local-walk}.

\begin{figure}
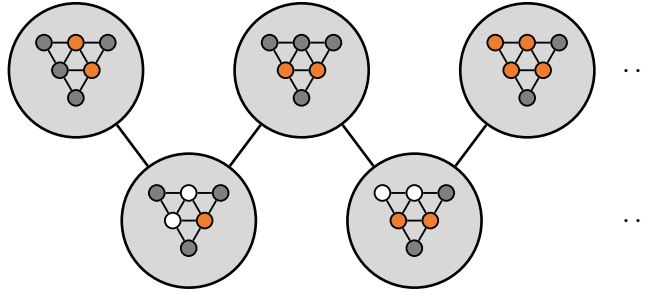

	\Tikz*{
		\tikzset{nd-/.style={fill=White, draw=Black, line width=1}}
		\tikzset{nd0/.style={fill=Gray, draw=Black, line width=1}}
		\tikzset{nd1/.style={fill=Orange, draw=Black, line width=1}}
		\tikzset{pics/graph/.style args={#1#2#3#4#5#6}{code={
			\begin{scope}[scale=0.7, every node/.style={circle, fill=White, inner sep=3}]
				\node[#1] (A) at (-30: 0.5) {};
				\node[#2] (B) at (90: 0.5) {};
				\node[#3] (C) at (-150: 0.5) {};
				\node[#4] (D) at (150: 1) {};
				\node[#5] (E) at (-90: 1) {};
				\node[#6] (F) at (30: 1) {}; 
			\end{scope}
			\foreach \u/\v in {A/B, B/C, C/A, D/B, D/C, E/A, E/C, F/A, F/B}
				\draw[line width=1, Black] (\u) -- (\v);
		}}}
		
		\tikzset{pics/graph short/.style args={#1#2#3#4#5#6}{code={\pic {graph={nd#1}{nd#2}{nd#3}{nd#4}{nd#5}{nd#6}};
		}}}
		
		\tikzset{bubble/.style={
			draw=Black, circle, fill=LightGray, line width=1
		}}
		
		\pgfmathsetmacro{\scl}{0.7}
		\pgfmathsetmacro{\l}{1}
		\pgfmathsetmacro{\v}{2}
		
		\node[bubble] (n0) at (-3*\l, 0) {\scalebox{\scl}{\Tikz{\pic{graph short=110000};}}};
		\node[bubble] (m0) at (-1.5*\l, -\v) {\scalebox{\scl}{\Tikz{\pic{graph short=1--000};}}};
		\node[bubble] (n1) at (0, 0) {\scalebox{\scl}{\Tikz{\pic{graph short=101000};}}};
		\node[bubble] (m1) at (1.5*\l, -\v) {\scalebox{\scl}{\Tikz{\pic{graph short=1-1-00};}}};
		\node[bubble] (n2) at (3*\l, 0) {\scalebox{\scl}{\Tikz{\pic{graph short=111100};}}};
		\node at (4.5*\l, 0) {$\dots$};
		\node at (4.5*\l, -\v) {$\dots$};
		\foreach \u/\v in {0/0, 1/0, 1/1, 2/1}
			\draw[line width=1] (n\u) -- (m\v);
	}
	\caption{The $k\leftrightarrow(k-2)$ random walk on monomers avoids the parity issue. In each round two vertices can change their membership in the monomer set. This is an instance of the \emph{multi-site} Glauber dynamics.}\label{fig:down-up}
\end{figure}

Unlike prior applications of spectral independence, we have to consider the $k\leftrightarrow \l$ random walk when $k-\l>1$. For example, consider the distribution of monomers in a monomer-dimer system. As we have established, we view this distribution on the ground set $\binom{V\times\set{0, 1}}{\card{V}}$, where $V$ is the set of vertices. The $k\leftrightarrow(k-1)$ random walk then becomes the following procedure, known as the (single-site) Glauber dynamics:
\begin{Algorithm*}
	Start with monomer set $S_0$\;
	\For{$t=0,1,\dots$}{
		Select vertex $v\in V$ uniformly at random\;
		Select $S_{t+1}$ between $S_t-\set{v}$ and $S_t\cup\set{v}$ randomly with probability $\propto \mu(\text{resulting set})$\;
	}
\end{Algorithm*}
It is not hard to see that cardinality of all monomer sets in a graph has a constant parity. This means that there is no transition possible from a monomer set $S$ to another set $S'$ that differs in exactly one vertex from it. Therefore the $k\leftrightarrow (k-1)$ walk produces a constant sequence $S_0, S_1=S_0,\dots$ and obviously does not mix. Note, however, that considering a higher value of $k-\l$ gets around this parity issue, see \cref{fig:down-up}.

We show that fractional log-concavity, and consequently, sector-stability, imply rapid mixing of the $k\leftrightarrow \l$ random walk for \emph{some} $\l=k-O(1)$. The following is the result of slight modifications of arguments by \textcite{AL20}.
\begin{theorem}
	Suppose that $\mu:\binom{[n]}{k}\to\R_{\geq 0}$ has a fractionally log-concave polynomial with parameter $\alpha=\Omega(1)$. Then for some $\l=k-O(1)$, the $k\leftrightarrow \l$ random walk started at the set $S_0$, gets $\epsilon$-close in total variation distance to the distribution $\mu$ in time
	\[ \tmix(\epsilon)=O\parens*{k^{O(1)}\cdot \log\parens*{\frac{1}{\epsilon\cdot \P_{\mu}{S_0}}}}. \]
\end{theorem}
One has to be careful that $\log(1/\P_\mu{S_0})$ is not too large in applications. This is achieved by making sure that $S_0$ has at least a $2^{-\poly(n)}$ probability under $\mu$. In all distributions we study in this paper, this can be achieved easily. For example, in the case of monomer-dimer distributions, by running a maximum-weight matching algorithm, we can find a matching $M$ having the maximum possible weight under the monomer-dimer distribution. Because the number of matchings is at most $2^{\poly(n)}$, we can safely use the monomer set of this matching as the starting point $S_0$.

\subsection{Acknowledgements}

We thank Micha\l{} Derezi\'{n}ski and Paul Liu for illuminating discussions about existing results on determinantal point processes. We also thank Alexander Barvinok for pointing us to existing results related to sector-stability.

Nima Anari and Thuy-Duong Vuong are supported by NSF grant CCF-2045354.

	\section{Preliminaries}
\label{sec:prelim}

We use $\Z_{\geq 0}$ to denote the set of nonnegative integers $\set{0,1,\dots}$. For a subset $S$ of $\R^n$, we use $\conv(S)$ to denote the convex hull of $S$.

We use $[n]$ to denote $\set{1,\dots,n}$. For a set $U$ we let $\binom{U}{k}$ denote the family of $k$-element subsets of $U$. When $n$ is clear from context, we use $\1_S\in \R^n$ to denote the indicator vector of the set $S\subseteq [n]$, having a coordinate of $0$ everywhere, except for elements of $S$, where the coordinate is $1$.

\subsection{Markov Chains}

For two measures $\mu, \nu$ defined on the same state space $\Omega$, we define their total variation distance as
\[ \dtv(\mu, \nu)=\frac{1}{2}\sum_{\omega\in \Omega}\abs{\mu(\omega)-\nu(\omega)}=\max\set{\P_\mu{S}-\P_\nu{S}\given S\subseteq \Omega}. \]

The total variation distance is a special case of a more general class of ``distance measures'' called $f$-divergences.
\begin{definition}[$f$-Divergence]
	For a convex function $f:\R_{\geq 0}\to \R$, define the $f$-divergence between two distributions $\mu$ and $\nu$ on the same state space as follows:
	\[ \D_f{\nu\river \mu}=\E*_{\omega\sim \mu}{f\parens*{\frac{\nu(\omega)}{\mu(\omega)}}}-f\parens*{\E*_{\omega\sim \mu}{\frac{\nu(\omega)}{\mu(\omega)}}}. \]
\end{definition}
Note that by Jensen's inequality this quantity is always nonnegative. Also notice that if $\mu$ and $\nu$ are normalized distributions the second term is just $f(1)$. In this work we will mostly deal with the case of $f(x)=x^2$, where $\D_f{\cdot\river\cdot}$ is also known as the variance. However, we state some results in full generality in terms of arbitrary $f$-divergences, in the hope that they fill find use in future work.

A Markov chain on a state space $\Omega$ is defined by a row-stochastic matrix $P\in \R^{\Omega\times \Omega}$. We view distributions $\mu$ on $\Omega$ as row vectors, and as such $\mu P$ would be the distribution after one transition according to $P$, if we started from a sample of $\mu$. A stationary distribution $\mu$ for the Markov chain $P$ is one that satisfies $\mu P=\mu$. Under mild assumptions on $P$ (ergodicity), stationary distributions are unique and the distribution $\nu P^t$ converges to this stationary distribution as $t\to \infty$ \cite{LP17}. We refer the reader to \cite{LP17} for a detailed treatment of Markov chain analysis.

A popular method for the analysis of Markov chains is via functional inequalities, that are often inequalities relating $f$-divergences before and after one transition of the Markov chain. We are specifically interested in contraction of the $f$-divergence. We state this contraction for (potentially non-square) row-stochastic operators for generality.
\begin{definition}
	We say that a row-stochastic matrix $P\in \R^{\Omega\times \Omega'}$ contracts $f$-divergence w.r.t.\ a background distribution $\mu:\Omega\to\R_{\geq 0}$ by a factor of $\alpha$ if for all other distributions $\nu:\Omega\to\R_{\geq 0}$, we have
	\[ \D_f{\nu P\river \mu P}\leq \alpha\cdot \D_f{\nu\river \mu}. \]
\end{definition}
We remark that all row-stochastic operators $P$ have contraction with factor $1$, and this property is only useful for $\alpha<1$.
\begin{proposition}[Data Processing Inequality]\label{prop:data-processing}
	For all row-stochastic matrices $P\in \R^{\Omega\times \Omega'}$ and all distributions $\mu, \nu:\Omega\to \R_{\geq 0}$, we have
	\[  \D_f{\nu P\river \mu P}\leq \D_f{\nu\river \mu}. \]
\end{proposition}

For a Markov chain $P$, we define the mixing time from a starting distribution $\nu$ as the first time $t$ such that $\nu P^t$ gets close to the stationary distribution $\mu$.
\[ \tmix(P, \nu, \epsilon)=\min\set{t\given \dtv(\nu P^t, \mu)\leq \epsilon}. \]
We drop $P$ and $\nu$ if they are clear from context. If $\nu$ is the Dirac measure on a single point $\omega$, we write $\tmix(P, \omega, \epsilon)$ for the mixing time. When mixing time is referenced without mentioning $\epsilon$, we imagine that $\epsilon$ is set to a reasonable small constant (such as $1/4$). This is justified by the fact that the growth of the mixing time in terms of $\epsilon$ can be at most logarithmic \cite{LP17}.

Contraction inequalities, combined with companion inequalities relating $\dtv$ and $f$-divergences allow one to bound the mixing time of a Markov chain. In particular for $f(x)=x^2$, one has the relationship
\[ \dtv(\nu, \mu)\leq O\parens*{\sqrt{\D_{x^2}{\nu\river \mu}}}, \]
and as a result we get
\begin{proposition}[{\cite[see, e.g.,][]{LP17}}]
	Suppose that a Markov chain $P$ with stationary distribution $\mu$ has $\alpha$-factor contraction in $\D_{x^2}{\cdot\river\cdot}$. Then the mixing time of $P$ started from a point $\omega$ satisfies
		\begin{align*}
	    \tmix(P, \nu, \epsilon)&\leq O\parens*{\frac{\log(1/\epsilon\P_{\mu}{\omega})}{\log(1/\alpha)}}\\
	    &\leq O\parens*{\frac{1}{1-\alpha}\cdot \log\parens*{\frac{1}{\epsilon\cdot \P_\mu{\omega}}}}.
	\end{align*}
\end{proposition}

\subsection{Complex Analysis}

We use the following classic result from elementary complex analysis \cite[see, e.g.,][]{Lan13}.
\begin{lemma}[Schwarz's lemma] \label{lem:schwarz}
	Let $D=\set*{z\in \C\given \abs{z}<1}$ be the open unit disk in the complex plane $\C$ centered at the origin and let $f:D\to \C$ be a holomorphic map such that $f(0)=0$ and $\abs{f(z)}\leq 1$ on $D$. Then
	\[ \abs{f'(0)}\leq 1. \]
\end{lemma}

\subsection{Linear Algebra}
	\begin{theorem}[Courant-Fischer Theorem]\label{thm:Courant-Fischer}
		Let $A\in \R^{n\times n} $ be a Hermitian matrix with eigenvalues $\lambda_1 \geq \lambda_2 \geq \cdots \geq \lambda_n.$
		Then
		\begin{align*}
		    \lambda_{k}(A) &= \adjustlimits\min_{U } \max_{v } \ \dotprod{v, Av},
		\end{align*}
		where the minimum is taken over all $(n-k+1)$-dimensional subspaces $U\subseteq \R^n$ and the maximum is 
		taken over all vectors $v\in U$ with $\dotprod{v,v}= 1$.
	\end{theorem}
	\begin{theorem}
	Let $A \in \R^{n\times m}, B \in \R^{m\times n}$ where $m \geq n$. Then the spectrum of BA (as a multiset) is precisely the
union of the spectrum of $AB$ (as a multiset) with $m- n$ copies of $0$.
	\end{theorem}

\subsection{Polynomials and Sector-Stability}

We use $\F[z_1,\dots,z_n]$ to denote $n$-variate polynomials with coefficients from $\F$, where we usually take $\F$ to be $\R$ or $\C$. We denote the degree of a polynomial $g$ by $\deg(g)$. We call a polynomial homogeneous of degree $k$ if all nonzero terms in it are of degree $k$. We define a $\lambda$-scaling, or an external field of $\lambda\in \F^n$ applied to a polynomial $g$, to be the polynomial $g(\lambda_1z_1,\dots,\lambda_nz_n)$. If $g$ was the generating polynomial of a distribution $\mu$, we denote the same scaling applied to $\mu$ by $\ef{\lambda}{\mu}$.

The main workhorse behind our main results are polynomials that avoid roots in certain regions of the complex plane.

\begin{definition}[Stability]
	For an open subset $U\subseteq \C^n$, we call a polynomial $g\in \C[z_1,\dots,z_n]$ $U$-stable iff
	\[ (z_1,\dots,z_n)\in U\implies g(z_1,\dots,z_n)\neq 0. \]
	We also call the identically $0$ polynomial $U$-stable. This ensures that limits of $U$-stable polynomials are $U$-stable. For convenience, when $n$ is clear from context, we abbreviate stability w.r.t.\ regions of the form $U\times U\times \cdots \times U$ where $U\subseteq \C$ simply as $U$-stability.
\end{definition}

Our choice of the region $U$ in this work is the product of open sectors in the complex plane. 

\begin{definition}[Sectors]
	We name the open sector of aperture $\alpha\pi$ centered around the positive real axis $\Gamma_\alpha$:
	\[ \Gamma_\alpha:=\set{\exp(x+iy)\given x\in \R, y\in (-\alpha\pi/2,\alpha\pi/2)}. \]
\end{definition}

With these definitions \cref{def:sector-stability} is the same as $\Gamma_\alpha$-stability for a suitable parameter $\alpha$.

Note that $\Gamma_1$ is the right-half-plane, and $\Gamma_1$-stability is the same as the classically studied Hurwitz-stability \cite[see, e.g.,][]{Bra07}. Another closely related notion is that of real-stability where the region $U$ is the upper-half-plane $\set{z\given \Im(z)>0}$ \cite[see, e.g.,][]{BBL09}. Note that for \emph{homogeneous} polynomials, stability w.r.t.\ $U$ is the same as stability w.r.t.\ any rotation/scaling of $U$; so Hurwitz-stability and real-stability are the same for \emph{homogeneous} polynomials.

\subsection{Half-Plane Stability}
% 	\begin{definition}[Real-stable]
% 	\end{definition}
% 	\begin{definition}[Hurwitz stable]
% 	\end{definition}
	%

	Consider an open half-plane $H_{\theta} = \set*{e^{-i\theta} z \given \Im(z) > 0} \subseteq \C.$ A polynomial $g(z_1,\cdots, z_n) \in \C[z_1, \cdots, z_n]$ is $H_{\theta}$-stable if $g$ does not have roots in $H_{\theta}^n.$
	We call $H_0$ and $H_{\pi/2}$ the upper half-plane and right half-plane respectively. 
	We say $g$ is Hurwitz-stable if it is $H_{\pi/2}$-stable. We say $g$ is real-stable if it is $H_0$-stable and has real coefficients.
	
	We observe that for homogeneous polynomials, the definition of $H_{\theta}$-stable is equivalent for all angles $\theta.$ %iff it is $H_{\theta'}$-stable, for any $\theta, \theta'.$ 
		\begin{lemma}[Theorem 1.6,  \cite{BB09}]\label{lemma:limitRoot}
	%TODO: Import Lemma 2.3 from \cite{BB09}. 
	Suppose that $f_j \in \C[z_1, \cdots , z_n]$ for all $j \in \N$ is $U$-stable for an
open set $U \subseteq \C^n$ and that $f$ is the limit, uniformly on compact subsets of $U$, of the sequence
$\set*{f_j}_{j\in \N}.$ Then $f$ is either $U$-stable or it is identically equal to 0.

In particular, if $f_j$ has bounded degree for all $j\in \N$, and the sequence $\set*{f_j}_{j\in \N}$ converges to $f$ coefficient-wise, then $f_j$ converge to $f$ uniformly on all compact sets in $\R^n$.
	\end{lemma}
	
	\begin{proposition}[Polarization, \cite{BBL09}] \label{prop:polarization} %Proposition 3.1
    For an element $\kappa $ of $\N^n$ let
    \[\R_{\kappa}[z_1, \cdots, z_n] = \set*{\text{polynomials in $\R[z_i]_{1\leq i\leq n}$ of degree at most $\kappa_i$
in $z_i$ for all $i$}} \]
\[\R_{\kappa}^a[z_{ij}] = \set*{\text{multi-affine polynomials in $\R[z_{ij}]_{1\leq i\leq n, 1 \leq j \leq \kappa_i}$}}\]
    The polarization map $\prod^{\uparrow}_{\kappa}$ is a linear map that sends monomial $z^{\alpha} = \prod_{i=1}^n z_i^{\alpha_i}$ to the product
    \[\frac{1}{\binom{\kappa}{\alpha}} \prod_{i=1}^n (\text{elementary symmetric polynomial of degree $\alpha_i$
in the variables $\set{z_{ij}}_{1\leq j \leq \kappa_i}$}
)\]
where $\binom{\kappa}{\alpha} = \prod_{i=1}^n \binom{\kappa_i}{\alpha_i}.$

A polynomial $g \in \R_{\kappa}[z_i]_{1\leq i\leq n}$ with nonnegative coefficients is real-stable if an only if its polarization $\prod^{\uparrow}_{\kappa}(g)$ is also real-stable. 
\end{proposition}
Taking polarization of $z^k$ with $\kappa = n$, we obtain the following well-known result.
	\begin{corollary} \label{thm:symmetric} For $k \leq n$, the $k$-th symmetric polynomial in $n$ variables $e_k(z_1, \cdots, z_n)$ is real-stable/Hurwitz-stable.
	\end{corollary}
	The following theorems will be useful in the proof of \cref{thm:DPPsample}.
	\begin{theorem} \label{thm:hurwitzSameParityPart}
	Let $g(z_1, \cdots, z_n) \in \R[z_1, \cdots, z_n]$ be Hurwitz-stable. Let $g_e$ ($g_o$) be the even (odd) part of $g$ i.e., the sum of terms $c_{\alpha} z^{\alpha}$ whose total degree $\abs{\alpha}_1$ is even (odd resp.). Then $g_e$ and $g_o$ are either identically $0$ or Hurwitz-stable.
	\end{theorem}
	\begin{proof}
	We have $g= g_e + g_o.$
	Replace $z_j$ with $iy_j$ with $y_j \in \mathcal{H}_0.$ Let $h(\set{y_j}_{j=1}^n) : = g(\set{i y_j}_{j=1}^n), h_e (\set{y_j}_{j=1}^n) := g_e (\set{i y_j}_{j=1}^n)$ and $h_o(\set{y_j}_{j=1}^n) := i^{-1} g_o (\set{i y_j}_{j=1}^n) $  then $h_e, h_o$ are polynomials with real coefficients, and $h$ is upper half-plane stable.
	
	We have $h = h_e + i h_o  $, and this is the unique way to write $h$ as $h_1 + ih_2$ where $h_j$ are polynomial with real coefficients, for $j \in \set*{1,2}.$ By \cite[Lemma 1.8, part (2)]{BB09}, $h_e$ and $h_o$ are real-stable or identically 0. Thus $g_e$, $g_o$ are Hurwitz-stable or identically 0. 
	%If $g_o \equiv 0$ then $g \equiv g_e$ so the claim is trivially true. Similarly, if $g_e\equiv 0$ then we are done.
	\end{proof}
	\begin{theorem}[\cite{BBL09}, Proposition 3.2] \label{thm:detMatrix}
	Let $A_1, \cdots, A_n$ be (complex) positive semi-definite matrices and
let $B$ be a (complex) Hermitian matrix, all matrices being of the same size $m\times m$.
\begin{enumerate}%[label = (\roman*)]
    \item 
    The polynomial
\[f(z_1 , \cdots, z_n) = \det(z_1 A_1 + \cdots + z_n A_n + B)\] 
is either identically zero or real-stable;
\item 
If $B$ is also positive semi-definite then $f$ has all non-negative coefficients.
\end{enumerate}
% In particular, if $Z = \diag{z_1, \cdots , z_n}$ and $A$ is a positive semi-definite $n\times n$
% matrix then $det(A+Z)$ is a multi-affine real-stable polynomial with all non-negative
% coefficients
	\end{theorem}
	%The following statements will be useful for proof of \cref{thm:DPPsample} and \cref{thm:monomerDimer}.
	
	\begin{lemma}%[determinant process by $p_0$ matrix is Hurwitz-stable] 
	\label{lemma:p0MatrixHurwitz}
	Consider $A \in \R^{n\times n}$ satisfying $A+A^T$ is positive semi-definite. %$A = D + X $ where $D\in \R^{n\times n}$ is positive semi-definite and $X\in \R^n$ is skew symmetric. %[$p_0$-matrix definition, see Lemma 1, \cite{Gartrell2019LearningND}]
	Let $f (z_1, \cdots, z_n) = \sum_{S \subseteq [n]} z^{[n] \setminus S} \det(A_{S,S}).$ Then $f$ has non-negative coefficients, and is either identically $0$ or Hurwitz-stable.
	\end{lemma}
	\begin{proof}
	Clearly, $A + A^T  $ is positive semi-definite, so $A$ is a $P_0$-matrix (see \cite[Lemma 1]{Gartrell2019LearningND}) i.e., all principle minors of $A$ are nonnegative.
	The coefficients of $f$ are principle minors of $A$, and are thus nonnegative.
	
	Let $D = (A + A^T)/2, X = (A - A^T)/2. $ Note that $X$ is skew-symmetric, thus
    $B:= iX$ is a Hermitian matrix, and $D$ is positive semi-definite.
	Apply \cref{thm:detMatrix} with $A_j = \diag{\textbf{e}_j}$ for $j\in [n]$ where $\textbf{e}_j$ is the $j$-th standard basis vector, $A_{n+1} = D$  and $B = i X$, we have
	$g(z_1, \cdots, z_n, z_{n+1}) :=\det(\sum_{i=1}^n z_i A_i + z_{n+1}D + iX) $ is either identically $0$ or real-stable. %as a polynomial in $(z_1, \cdots, z_n, z_{n+1})$.
	
	Let $w_j := i^{-1} z_j$, $Z = \sum_{i=1}^n z_i A_i = \diag{z_1, \cdots, z_n}$ and $W = \diag{w_1, \cdots, w_n }$. We can rewrite
	\begin{align*}
	    g(z_1, \cdots, z_n, i) &= \det(Z + iD +iX) = \det(iW + iA) = i^{n} \det(W +A) \\
	    &=  i^n \sum_{S \subseteq [n]} w^{[n] \setminus S} \det(A_{S,S})  = i^n f(w_1, \cdots, w_n)
	\end{align*}
	If $g\equiv 0$ then so is $f.$ Suppose $g\not\equiv 0.$ 
	Fix arbitrary $w_1, \cdots,w_n $ in the right half plane $H_{\pi/2}$. Observe that $z_j = iw_j$ is in the upper half plane $H_0$. Real-stability of $g$ implies $f(w_1, \cdots, w_n) = g(z_1, \cdots, z_n, i)\neq 0 ,$ hence $f$ is Hurwitz-stable.  
	\end{proof}
	We also need the following for the proof of \cref{thm:monomer-dimer}.

	\begin{theorem}[\cite{HL72}] \label{thm:monomerDimerPoly}
	Consider a graph $G = G(V,E)$ on $n$ vertices with edge weight $w : E \to \R_{\geq 0}$ and vertex weight $\lambda: V \to \R_{\geq 0}$. For $S \subseteq V$, let $m_S := \sum_M \weight(M) =  \sum_M(\prod_{e\in M} w(e) \prod_{v\not\in S} \lambda(v))$ where the sum is taken over all perfect matchings $M$ of $S$. 
	The following polynomial is Hurwitz-stable
	\[f(z_1, \cdots, z_n) = \sum_{S \subseteq V} z^{[n] \setminus S} m_S \]
	\end{theorem}

	\subsection{Matroids}
	A matroid $M = (E,\I)$ is a structure consisting of a finite ground set $E$ and a non-empty collection $\I$ of \emph{independent} subsets of $E$ satisfying: 
\begin{enumerate}  
	\item If $S \subseteq T$ and $T \in \I$, then $S \in \I$.  
	\item If $S,T\in \I$ and $\card{T} > \card{S}$, then there exists an element $i \in T\setminus S$ such that $S\cup \set{i} \in \I$. 
\end{enumerate}
The \emph{rank} of a matroid is the size of the largest independent set of that matroid. If $M$ has rank $r$, any set $S\in \I$ of size $r$ is called a \emph{basis} of $M$. Let $\B_M \subset \I$ denote the set of bases of $M$. The set of bases $\B_M$ of a matroid unique define $M.$

We say a matroid $M$ is strongly Rayleigh or satisfies the weak half-plane property if $f(z_1, \cdots, z_n) = \sum_{S \in \B_{M}} z^{S}$ is real-stable. 

For partition $T_1, \cdots, T_s$ of $[n]$, and tuple $(c_1, \cdots, c_s) \in \N^s$, the partition matroid $M$ associated with $(T_1, \cdots, T_s)$ and $(c_1, \cdots, c_s)$ is defined by $\B_M = \set*{S\subseteq [n] \given \abs{S \cap T_i} = c_i \forall i}.$
	
	\section{Down-Up Random Walks and Spectral Independence}

Here we establish sufficient conditions for rapid mixing of the $k\leftrightarrow \l$ down-up random walks as defined in \cref{def:local-walk}.

\begin{remark}
	Our arguments in this section are small tweaks of the local-to-global contraction analyses already found in prior work of \textcite{AL20} and \textcite{CGM19}; the origin of these types of arguments goes back to the study of high-dimensional expanders  \cite{KM16,DK17,KO18}, and more sophisticated variants useful in the context of Markov chain analysis can be found in recent works of \textcite{CLV20a, CLV20b, GM20}. For the mixing time bounds in this work, the analysis of \textcite{AL20} and the framework built on it by \textcite{ALO20} dubbed ``spectral independence'' suffices; however, we choose to state a general local-to-global contraction analysis not found explicitly in prior work, in the hope that it will find use in future applications.
\end{remark}

For a distribution $\mu:\binom{[n]}{k}\to \R_{\geq 0}$, our goal is to analyze the mixing time of the $k\leftrightarrow \l$ down-up random walk. We will do this by establishing contraction of $f$-divergence in these random walks. Similar to prior results on local-to-global analysis of high-dimensional expanders, our goal is to show that ``local'' contraction of $f$-divergence (where the down-up walks are applied to a ``localization'' of $\mu$) implies ``global'' contraction of $f$-divergence.

The down-up walks can be written as the composition of two row-stochastic operators known aptly as the down and up operators.
\begin{definition}[Down Operator]
	For a ground set $[n]$, and cardinalities $k\geq \l$ define the row-stochastic down operator $D_{k\to \l}\in \R^{\binom{[n]}{k}\times \binom{[n]}{\l}}$ as
	\[ 
		D_{k\to \l}(S, T)=\begin{cases}
			\frac{1}{\binom{k}{\l}}&\text{ if }T\subseteq S,\\
			0&\text{ otherwise}.\\
		\end{cases}
	\]
\end{definition}
This operator applied to a random set $S$, produces a uniformly random subset $T$ of size $\l$ out of it. The down operators compose in the way one expects, i.e., $D_{k\to \l}D_{\l\to m}=D_{k\to m}$. Note that the down operator has no dependence on $\mu$. In contrast the up operator as defined below depends on $\mu$ and is actually designed to be the time-reversal of the down operator w.r.t.\ the background measure $\mu$.
\begin{definition}[Up Operator]
	For a ground set $[n]$, cardinalities $k\geq \l$, and density $\mu:\binom{[n]}{k}\to \R_{\geq 0}$, define the up operator $U_{\l \to k}\in \R^{\binom{[n]}{\l}\times \binom{[n]}{k}}$ as
	\[ 
		U_{\l\to k}(T, S)=\begin{cases}
			\frac{\mu(S)}{\sum_{S'\supseteq T}\mu(S')}&\text{ if }T\subseteq S,\\
			0&\text{ otherwise}.\\
		\end{cases}
	\]
\end{definition}
If we name $\mu_k=\mu$ and more generally let $\mu_\l$ be $\mu_k D_{k\to \l}$, then the down and up operators satisfy detailed balance (time-reversibility) w.r.t.\ the $\mu_k, \mu_\l$ operators. In other words we have
\[ \mu_k(S)D_{k\to \l}(S, T)=\mu_\l(T)U_{\l \to k}(T, S). \]
This property ensures that the composition of the down and up operators have the appropriate $\mu$ as a stationary distribution, are time-reversible, and have nonnegative real eigenvalues.
\begin{proposition}[{\cite[see, e.g.,][]{KO18,AL20,ALO20}}]
	The operators $D_{k\to \l}U_{\l\to k}$ and $U_{\l\to k}D_{k\to \l}$ both define Markov chains that are time-reversible and have nonnegative eigenvalues. Moreover $\mu_k$ and $\mu_\l$ are respectively their stationary distributions.
\end{proposition}

Our goal is to show that these operators contract $f$-divergence by a multiplicative factor. To this end, it is enough to show contraction of $f$-divergence under $D_{k\to \l}$. This is because, by the data processing inequality, \cref{prop:data-processing}, the operator $U_{\l\to k}$ cannot increase the $f$-divergence.

The key ingredient in local-to-global arguments is the ``local contraction'' assumption. Here, one assumes that $D_{2\to 1}$ contracts $f$-divergences w.r.t.\ the background measure $\mu_2$. The goal is to go from this assumption, and similar ones for conditionings of $\mu$, see \cref{def:conditioned}, to contraction of $f$-divergence for $D_{k\to \l}$. This is the natural ``$f$-divergence'' generalization of the notion of local spectral expansion and its implications for global expansion \cite[see][]{KO18}.

First we define the notion of the link of the distribution $\mu$ w.r.t.\ a set $T$ \cite[see, e.g.,][]{KO18}. This notion is almost the same as the notion of conditioned distributions $\mu_T$, see \cref{def:conditioned}, except we remove the set $T$ as well.
\begin{definition}
	For a distribution $\mu:\binom{[n]}{k}\to\R_{\geq 0}$ and a set $T\subseteq [n]$ of size at most $k$, we define the link of $T$ to be the distribution $\mu_{-T}:\binom{[n]-T}{k}\to \R_{\geq 0}$ which describes the law of the set $S-T$ where $S$ is sampled from $\mu$ conditioned on the event $S\supseteq T$.
\end{definition}

Next we define the notion of local $f$-divergence contraction for a distribution $\mu$.
\begin{definition}[Local $f$-Divergence Contraction]
	For a distribution $\mu:\binom{[n]}{k}\to\R_{\geq 0}$ and a set $T$ of size at most $k-2$, define the local contraction at $T$, to be the smallest number $\alpha(T)\geq 0$ such that $D_{2\to 1}$ contracts $f$-divergences w.r.t.\ $\parens*{\mu_{-T}}_2=\mu_{-T}D_{(k-\card{T})\to 2}$ by a factor of $\alpha(T)$. That is $\alpha(T)$ is the smallest number such that for all $\nu:\binom{[n]-T}{2}\to \R_{\geq 0}$ we have
	\[ \D*_f{\nu D_{2\to 1}\river \pi_T(\mu)D_{(k-\card{T})\to 1}}\leq \alpha(T)\cdot \D*_f{\nu\river \pi_T(\mu)D_{(k-\card{T})\to 2}}. \]
\end{definition}

We now show that local contraction of $f$-divergence results in a bound on the contraction of $D_{k\to \l}$ operators.
\begin{theorem}
	Suppose that $\mu:\binom{[n]}{k}\to \R_{\geq 0}$ has local $f$-divergence contraction with contraction factors $\alpha(T)$. Define $\beta(T)=\min\set{1,\alpha(T)/(1-\alpha(T))}$. For a set $T\subseteq [n]$ define
	\[ \gamma_T:=\E*_{e_1,\dots,e_m\sim \text{uniformly random permutation of }T}{\beta({\emptyset})\beta(\set{e_1})\cdots \beta(\set{e_1,\dots, e_{m}})}. \]	
	Then the operator $D_{k\to \l}$ has contraction factor at least $1-1/\max\set*{k\cdot \gamma_T\given T\in \binom{[n]}{\l-1}}$.
\end{theorem}
\begin{proof}
	Consider an arbitrary distribution $\nu:\binom{[n]}{k}\to\R_{\geq 0}$. The $f$-divergence $\D_f{\nu\river \mu}$ is a difference of two terms, both involving expectations over samples $S\sim \mu$:
	\[ \D_f{\nu\river \mu}=\E*_{S\sim \mu}{f\parens*{\frac{\nu(S)}{\mu(S)}}}-f\parens*{\E*_{S\sim \mu}{\frac{\nu(S)}{\mu(S)}}}. \]
	Our strategy is to write this difference as a telescoping sum of differences, where elements of $S$ are revealed one-by-one in the sum.
	
	Consider the following process. We sample a set $S\sim \mu$ and uniformly at random permute its elements to obtain $X_1,\dots,X_k$. Define the random variable
	\[ \tau_i=f\parens*{\E*{\frac{\nu(S)}{\mu(S)}\given X_1,\dots,X_i}}=f\parens*{\frac{\sum_{S'\ni X_1,\dots, X_i}\nu(S')}{\sum_{S'\ni X_1,\dots,X_i}\mu(S')}}=f\parens*{\frac{\nu D_{k\to i}(\set{X_1,\dots,X_i})}{\mu D_{k\to i}(\set{X_1,\dots,X_i})}}. \]
	Note that $\tau_i$ is a ``function'' of $X_1,\dots,X_i$. It is not hard to see that 
	\[ \D_f{\nu\river \mu}=\E{\tau_k}-\E{\tau_0}=\sum_{i=0}^{k-1}(\E{\tau_{i+1}}-\E{\tau_{i}}). \]
	A convenient fact about this telescoping sum is that to obtain $\D_f{\nu D_{k\to \l}\river \mu D_{k\to\l}}$, one has to just sum over the first $\l$ terms instead of $k$:
	\[ \D_f{\nu D_{k\to \l}\river \mu D_{k\to \l}}=\E{\tau_\l}-\E{\tau_0}=\sum_{i=0}^{\l-1}(\E{\tau_{i+1}}-\E{\tau_{i}}). \]
	This is because the set $\set{X_1,\dots,X_\l}$ is distributed according to $\mu D_{k\to \l}$. So our goal of showing that $D_{k\to\l}$ has contraction boils down to showing that the last $k-\l$ terms in the telescoping sum are sufficiently large compared to the rest.
	
	Consider applying the assumption of local contraction to the link of the set $T=\set{X_1,\dots,X_i}$. From this one can extract that
	\[ \E{\tau_{i+1}\given X_1,\dots,X_i}-\E{\tau_i\given X_1,\dots,X_i}\leq \alpha(T)\cdot (\E{\tau_{i+2}\given X_1,\dots,X_i}-\E{\tau_i\given X_1,\dots, X_i}). \]
	Defining $\Delta_i=\tau_{i+1}-\tau_i$, the above can be rewritten as
	\[ \E{\Delta_{i}\given X_1,\dots,X_i}\leq \alpha(\set{X_1,\dots,X_i})\cdot \E{\Delta_{i}+\Delta_{i+1}\given X_1,\dots, X_i}. \]
	Rearranging yields
	\[ \E{\Delta_{i}\given X_1,\dots,X_i}\leq \frac{\alpha(\set{X_1,\dots,X_i})}{1-\alpha(\set{X_1,\dots,X_i})}\E{\Delta_{i+1}\given X_1,\dots,X_i}\leq \beta(\set{X_1,\dots,X_i})\E{\Delta_{i+1}\given X_1,\dots,X_i}. \]
	From this we obtain that if we consider the quantities
	\[ \Delta_i\cdot \beta({\emptyset})\cdot \beta(\set{X_1})\cdots \beta(\set{X_1,\dots,X_{i-1}}), \]
	they form a submartingale; this means that we have
	\[ \E*{\Delta_\l \cdot \beta(\emptyset)\cdots \beta(\set{X_1,\dots,X_{\l-1}})}\geq \E*{\Delta_0}. \]
	Now, consider an alternative process for generating the ordering $X_1,X_2,\dots,X_k$. First select $S\sim \mu$, and partition it into two sets, $T$ of size $\l-1$ and $S-T$ of size $k-\l+1$. We then randomly shuffle $T$ and let $X_1,\dots,X_{\l-1}$ be the result, and then randomly shuffle $S-T$ and let $X_{\l},\dots,X_k$ be the result. This process is equivalent to randomly shuffling all elements of $S$.
	
	The key insight is that $\Delta_\l$ is only a function of the \emph{unordered set} $T$ and the ordering of $S-T$. However the other factor $\beta(\emptyset)\cdots \beta(\set{X_1,\dots,X_{\l-1}})$ is only a function of the ordering chosen for $T$ and not $S-T$. This means that conditioned on $T$, these two quantities are independent and we get
	\[ \E*{\Delta_\l \cdot \beta(\emptyset)\cdots \beta(\set{X_1,\dots,X_{\l-1}})}=\E*_T{\E*{\Delta_\l\given T}\cdot \E*{\beta(\emptyset)\cdots \beta(\set{X_1,\dots,X_{\l-1}})\given T}}. \]
	From the definition of $\gamma_T$, we obtain
	\[ \E*{\Delta_\l \cdot \beta(\emptyset)\cdots \beta(\set{X_1,\dots,X_{\l-1}})}\leq \E{\Delta_\l}\cdot \max\set*{\gamma_T\given T\in \binom{[n]}{\l-1}}. \]
	Combining with previous inequalities we obtain
	\[ \E*{\Delta_\l}\geq \frac{\E{\Delta_0}}{\max\set*{\gamma_T\given T\in \binom{[n]}{\l-1}}}. \]
	Similar inequalities can be obtained with $\Delta_0$ replaced by $\Delta_1,\Delta_2,\dots$ in the above arguments (with potentially better factors than $\gamma_T$, but we ignore this potential improvement). Averaging over these $k$ inequalities we obtain
	\[ \E{\Delta_\l}\geq \frac{\E{\Delta_0+\dots+\Delta_{k-1}}}{\max\set*{k\cdot \gamma_T\given T\in \binom{[n]}{\l-1}}}=\frac{\D_f{\nu\river \mu}}{\max\set*{k\cdot \gamma_T\given T\in \binom{[n]}{\l-1}}}. \]
	It just remains to note that
	\[ \D_f{\nu\river \mu}-\D_f{\nu D_{k\to \l}\river \mu D_{k\to \l}}=\E{\Delta_\l+\dots+\Delta_{k-1}}\geq \E{\Delta_\l}. \]
	Here we used nonnegativity of $\E{\Delta_i}$ which follows from convexity of $f$ and Jensen's inequality. Combining the previous two inequalities and rearranging the terms yields the desired result.
\end{proof}

\begin{remark}
We remark that similar to prior works, in this paper we only deal with the case where the $\alpha(T)$ contraction factors only depend on the size $\card{T}$. However, we suspect the more general statement we proved here to be useful in potential future applications of this method, especially to distributions $\mu$ that ``factorize'' into two independent distributions when conditioned on an element; some potential examples include distributions over chains in a poset. In these scenarios, the order of conditioning on the elements matters, and we hope that by having $\E_{\text{orderings}}{\beta({\emptyset})\beta(\set{e_1})\cdots \beta(\set{e_1,\dots, e_{m}})}$ instead of $\max_{\text{orderings}}\set{\beta({\emptyset})\beta(\set{e_1})\cdots \beta(\set{e_1,\dots, e_{m}})}$, we get more tractable results.
\end{remark}

From this point on, we deal with cases where $\alpha(T), \beta(T)$ only depend on the cardinality $\card{T}$, and as such we write them as $\alpha_i, \beta_i$, where $i=\card{T}$. Consequently, the global contraction factor we obtained can be rewritten as
\[ 1-\frac{1}{k\beta_0\beta_1\cdots \beta_{\l-1}}. \]
\begin{remark}
	A similar, slightly better, contraction factor can be obtained when $\beta(T)$ only depend on $\card{T}$. In these cases one can simply use $\E{\Delta_i}\leq \beta_i\cdot \E{\Delta_{i+1}}$ and obtain that the we have contraction
	\[ \frac{\E{\Delta_0+\dots+\Delta_{\l-1}}}{\E{\Delta_0+\dots+\Delta_{k-1}}}\leq \frac{1+1/\beta_0+\dots+1/\beta_0\cdots \beta_{\l-2}}{1+1/\beta_0+\dots+1/\beta_0\cdots \beta_{k-2}}. \]
	This is essentially the same bound found by \textcite{CLV20b,GM20} and the analysis is essentially the same as those in its core. However this slightly better bound does not produce any meaningful improvement in the mixing time bounds we get in this work, and for simplicity we use the more naive bound.
\end{remark}

While it might seem that $\beta_0\cdots \beta_{\l-1}$ can get exponentially large, in the case of distributions that satisfy spectral independence \cite{ALO20}, this product remains polynomially small. In particular, one can show \cite[see, e.g.,][]{ALO20,CLV20b} that if the correlation matrix, see \cref{def:correlation-matrix}, has $O(1)$-bounded eigenvalues for the distribution $\mu$ and all of its conditionings, then $\beta_i\simeq 1/(1-O(1/(k-i)))$. In particular, as long as $k-i$ is larger than a constant (hidden in the $O$-notation), then $\beta_i$ is finite an can be roughly approximated by $e^{O(1/(k-i))}$. Thus for $k-\l$ larger than an appropriate constant, we have the bound
\[ \beta_0\beta_1\cdots \beta_{k-\l}\simeq \exp\parens*{O\parens*{\frac{1}{k}+\frac{1}{k-1}+\dots+\frac{1}{\l}}}\leq \exp(O(\log k))=\poly(k). \]
	
	\section{Sector-Stability Implies Bounded Correlations}

In this section, we prove \cref{thm:sector-implies-bounded}.
	\begin{definition}[Signed Pairwise Influence/Correlation Matrix] \label{def:corr}
	Let $\mu$ be a probability distribution over $2^{[n]} $ with generating polynomial $f (z_1, \cdots, z_n) = \sum _{S \in  2^{[n]}}\mu(S) z^S.$
	
	Let the \textit{signed pairwise influence matrix} $\inflMat_{\mu} \in \R^{n\times n}$  be defined by
\[
\inflMat_{\mu} (i,j)  = \begin{cases}  0 &\text{ if } j=i \\ \P{j\given i} - \P{j \given \bar{i}} &\text{ else} \end{cases} \]
where $\P{j\given i} = \P_{T \sim \mu}{ j \in T \given i \in T}, \P{j} = \P_{T \sim \mu}{j \in T} $ and $ \P{j \given \bar{i}} = \P_{T \sim \mu}{ j \in T \given i \not\in T}.$

	Let the \textit{correlation matrix} $\corMat_{\mu} \in \R^{n\times n}$  be defined by
	\[\corMat_{\mu} (i,j) = \begin{cases} 1 -\P{i} &\text{ if } j=i \\ \P{j \given i} - \P{j} &\text{ else}\end{cases} \]
	\end{definition}
	
	In \cref{def:corr}, we use the convention that the entry $\inflMat(i,j)$ ($\corMat(i,j)$ resp.) is set to $0$ if $\P{j \given i}$ or $\P{j \given \bar{i}}$ ($\P{j \given i}$ resp.) are not well-defined, e.g., $\P{i} = 0$ or $\P{\bar{i}} = 0$ ($\P{i} = 0$ resp.)
	
	Note that the influence matrix, $\Psi^{\inf}_\mu$, was first introduced in \cite{ALO20}. All of  eigenvalues of $\inflMat_{\mu}$ and $\corMat_\mu$ are real \cite[see, e.g.,][]{ALO20}.	
	
 	%\begin{definition}[Correlation matrix]
	%\end{definition}
	We show that $\Omega(1)$-aperture sector-stability of the generating polynomial of $\mu$ implies $O(1)$-bound on the row norms of $\inflMat_{\mu}$ and $\corMat_{\mu}$. The high level idea is to write the $\l_1$-norm of a row of $\inflMat$ as the derivative at $0$ of some holomorphic function that maps the unit disk to itself, and then use Schwarz's Lemma (\cref{lem:schwarz}) to derive a bound.

	\begin{theorem}\label{thm:sectorStableRowNormBound}
Consider a multi-affine $ f \in \R_{\geq 0}[z_1, \cdots, z_n]$ polynomial that is $\Gamma_{\alpha}$-stable with $\alpha \leq 1$. Let $\mu: 2^{[n]} \to \R_{\geq 0}$ be the distribution generated by $f$, then $\inflMat_{\mu}$ and $\corMat_{\mu}$ have bounded row norms. Specifically, 
\[ \sum_{j}\abs{\inflMat_{\mu} (i,j)} \leq 2/\alpha -1,\]
and
\[ \sum_{j}\abs{\corMat_{\mu} (i,j)} \leq 2/\alpha. \]
As a corollary, the same bounds hold for maximum eigenvalues, i.e., $\lambda_{\max}(\inflMat_{\mu}) \leq  2/\alpha -1 $ and $\lambda_{\max}(\corMat_{\mu}) \leq  2/\alpha$.
\end{theorem}

\begin{proof}
If we can show the first statement, the second follows from
\[\P{j\given i} - \P{j} = \P{j\given i} - (\P{j \given i}\P{i} + \P{j \given \bar{i}}\P{\bar{i}}) = (1-\P{i})  (\P{j\given i} -\P{j \given \bar{i}}) \]
\[\sum_{j}\abs{\corMat_{\mu} (i,j)} \leq (1-\P{i}) (1+ \sum_{j\neq i} \abs{\P{j\given i} -\P{j \given \bar{i}}}) \leq 2/\alpha. \]

Fix a row $i$. W.l.o.g., assume $i=n.$ Let $h = \partial_i f, g = f_{z_i=0}.$ We can assume w.l.o.g. that neither $g$ and $h$ are the zero polynomial. If either $g$ or $h$ are the zero polynomial then the row just become identically $0$ and the statement is trivial. %Note that
% \[\P{j} = \frac{1}{d}\frac{\partial_j g (\vec{1}) }{ g(\vec{1}) }, \P{j \given i} = \frac{1}{d-1}\frac{\partial_j h (\vec{1}) }{ h(\vec{1}) }  \]
% We want to bound
% \begin{equation}
%     \begin{split}
%         \abs{M_i}_1 &= \sum_{j\neq i} \abs{\P{j \given i}  - \P{j} } = \sum_{j\neq i}\abs{\frac{1}{d-1}\frac{\partial_j h (\vec{1}) }{ h(\vec{1}) }  - \frac{1}{d}\frac{\partial_j g (\vec{1}) }{ g(\vec{1}) } } \\
%         &\leq \sum_{j\neq i} \frac{1}{d} \abs{\frac{\partial_j h (\vec{1}) }{ h(\vec{1}) }  - \frac{\partial_j g (\vec{1}) }{ g(\vec{1}) } }+ (\frac{1}{d-1} -\frac{1}{d}) \sum_{j\neq i} \frac{\partial_j h (\vec{1}) }{ h(\vec{1}) }   = \sum_{j\neq i} \frac{1}{d} \abs{\frac{\partial_j h (\vec{1}) }{ h(\vec{1}) }  - \frac{\partial_j g (\vec{1}) }{ g(\vec{1}) } } +\frac{1}{d}
%     \end{split}
% \end{equation}
%Let $d = \deg(f)$ then $\deg(g) = d$ and $\deg(h) = d-1.$
%Note that $\sum_{j\neq i} (\P{j \given \bar{i}} -\P{j\given i}) =\deg (g) - \deg(h) = 1$. 
Let $S: =\set*{ j \in [n]\setminus \{i\} \given \P{j \given i} - \P{j \given \bar{i}} < 0 } $ then
\begin{equation}\label{eq:rowsum}
   \sum_{j \neq i} \abs{\inflMat_{\mu}(i,j)} = \sum_{j\in S}   (\P{j \given \bar{i}} - \P{j \given i})  - \sum_{j \not\in S}   (\P{j \given \bar{i}} - \P{j \given i})  .%= 2 \sum_{j \in S} (\P{j \given \bar{i}} - \P{j \given i}) -1 
\end{equation}
Note that $\P{j \given i} = \frac{\partial_j h (\vec{1}) }{ h(\vec{1})}$ and $\P{j \given \bar{i}} = \frac{\partial_j g (\vec{1}) }{ g(\vec{1})}.$

Define $\vec{z} \in \R^{n-1}$ by $z_j =\begin{cases} y &\text{ for } z \in S \\ y^{-1} &\text{ else}\end{cases}.$ 

Let $\bar{h} (y) = h(\vec{z})$ and $\bar{g} (y) = g(\vec{z}).$ Note that  $\sum_{j\in S} \partial_j h (\textbf{1}) - \sum_{j\not\in S} \partial_j h (\textbf{1}) = \bar{h}'(1)$ and the same goes for $\bar{g}.$ This is because for each monomial $ z^{U} = z^{U\cap S} z^{U \setminus S} $, we have \[\left(\sum_{j\in S} \partial_j z^U - \sum_{j\not \in S} \partial_j z^U \right)_{\vec{z}=\textbf{1}}= \abs{U\cap S} -\abs{U \setminus S} = \left(y^{\abs{U\cap S}} (y^{-1})^{\abs{U \setminus S}}\right)'|_{y=1} \]
Substitute into \eqref{eq:rowsum}, we get
\begin{equation} \label{eq: rowsum to derivative}
    \sum_{j=1}^{n-1} \abs{\inflMat_{\mu}(i,j)} =\sum_{j\in S} (\frac{\partial_j g (\vec{1}) }{ g(\vec{1}) }  - \frac{\partial_j h (\vec{1}) }{ h(\vec{1}) }) - \sum_{j\not\in S} (\frac{\partial_j g (\vec{1}) }{ g(\vec{1}) }  - \frac{\partial_j h (\vec{1}) }{ h(\vec{1}) }) = \frac{\bar{g}'(1)}{\bar{g}(1)} - \frac{\bar{h}'(1)}{\bar{h}(1)} = (\log \bar{g} - \log \bar{h})'_{y=1} = \phi'(0)
\end{equation}
% \[\sum_{j=1}^{n-1} \abs{\P{j \given i} - \P{j \given \bar{i}}}=\sum_{j\in S} (\frac{\partial_j g (\vec{1}) }{ g(\vec{1}) }  - \frac{\partial_j h (\vec{1}) }{ h(\vec{1}) }) - \sum_{j\not\in S} (\frac{\partial_j g (\vec{1}) }{ g(\vec{1}) }  - \frac{\partial_j h (\vec{1}) }{ h(\vec{1}) }) = \frac{\bar{g}'(1)}{\bar{g}(1)} - \frac{\bar{h}'(1)}{\bar{h}(1)} = (\log \bar{g} - \log \bar{h})'_{y=1} = \phi'(0) \]
where $\phi(x) = \log \frac{\bar{g}(e^x)}{\bar{h}(e^x)} - \log \frac{\bar{g}(1)}{\bar{h}(1)}.$ Note that $\phi$ maps $0$ to itself.

Let $D, H\subseteq \C$ be the centered (open) unit disk and the (open) right half-plane respectively.  For any set $\Omega \subseteq \C$, we let $\overline{\Omega}$ denote its closure.

The Mobius transformation $T: x \mapsto \frac{x-1}{x+1}$ is a conformal map from $H$ onto $D$. %Moreover, $T$ maps the imaginary axis to the unit circle $\set*{z \in \C\given \abs{z} = 1},$ thus $T$ maps the closed right half plane $\overline{\mathbb{H}}$ to the closed unit disk $\overline{\mathbb{D}}.$

For angle $\theta \in (0,\pi)$ let $\Omega_{\theta}: = \set*{x \in \C \given  \abs{\Im(x)} < \theta } $ and  $\varphi_{\theta}: \Omega_{\theta} \to D, x \mapsto T(\exp(\frac{\pi x}{2\theta})). $  Note that $\varphi_{\theta}(0) = T(1) = 0,$ $\varphi_{\theta}'(0) = T'(1) \frac{\pi}{2 \theta} = \frac{\pi}{4 \theta} $ and $(\varphi_{\theta}^{-1})'(0) = \frac{1}{\varphi'_{\theta}(0) } =\frac{4 \theta}{\pi}. $ %Moreover, $\varphi_{\theta}$ maps the boundary $\set*{x\in \C\given \abs{\Im(x)} = \theta }$ of $\overline{\Omega_{\theta}}$ to the unit circle, thus $\varphi_{\theta}$ is 

To bound $\abs{\phi'(0)}$, we show that $\phi$ maps $\Omega_{\alpha\pi/2}$ to $\overline{\Omega_{\pi-\alpha\pi/2}}$. Now, for all small $\epsilon > 0$, $\tilde{\phi} :=  \varphi_{\pi-\alpha\pi/2 + \epsilon} \circ \phi \circ \varphi_{\alpha\pi/2}^{-1}$ is a holomorphic function that takes the centered unit disk to itself. We use Schwarz Lemma to bound $\abs{\tilde{\phi}'(0)},$ then use this to bound $\abs{\phi'(0)}$. 

Let $\theta :=\alpha \pi/2.$ Consider $x \in \Omega_{\theta}.$
Note that the function $x \mapsto e^x$ maps $\Omega_{\theta}$ to %$\theta$-sector 
$S_{\alpha}.$ Also, $\frac{\bar{g}(e^x)}{\bar{h}(e^x)} \not\in -S_{\alpha}$ else $\bar{g}(e^x) + \bar{h}(e^x) z = 0 $ for some $z \in S_{\alpha}$ i.e., $f(e^x,\cdots, e^x, z) = 0$, which contradicts $S_{\alpha}$-sector-stability of $f.$ In particular, $ \frac{\bar{g}(e^x)}{\bar{h}(e^x)}$ never takes negative real value, thus the function $\log \frac{\bar{g}(e^x)}{\bar{h}(e^x)}$ is holomorphic, and as argued earlier, $\abs{\Im(\log \frac{\bar{g}(e^x)}{\bar{h}(e^x)})} \leq \pi-\theta.$ Additionally, since $g, h $ has non-negative coefficients and are not the zero polynomial, $\bar{g}(1)$ and $\bar{h}(1)$ are positive real and $\log \frac{\bar{g}(1)}{\bar{h}(1)}$ is a real number. Therefore, $\abs{\Im (\phi(x))} = \abs{\Im(\log \frac{\bar{g}(e^x)}{\bar{h}(e^x)})} \leq \pi-\theta.$  Hence, $\phi $ maps $\Omega_{\theta}$ to $\Omega_{\pi-\theta + \epsilon}$ for every $\epsilon > 0.$

Fix $\epsilon > 0.$
Consider the holomorphic map $\tilde{\phi} =  \varphi_{\pi-\theta + \epsilon}\circ \phi\circ\varphi_{\theta}^{-1} $ that takes $D$ to itself. Since $\phi, \varphi_*$ both take $0$ to itself, so is $\tilde{\phi}.$ By Schwarz's Lemma (\cref{lem:schwarz}), $\abs{\tilde{\phi}'(0)} \leq 1.$ On the other hand, $ \tilde{\phi}'(0) = \varphi_{\pi-\theta + \epsilon}' (0)  \times \phi'(0) \times (\varphi_{\theta}^{-1})' (0) = \frac{\pi}{4 (\pi-\theta + \epsilon)} \phi'(0) \frac{4 \theta}{\pi} = \frac{\theta}{\pi-\theta + \epsilon} \phi'(0), $   thus $\abs{\phi'(0)} \leq \frac{\pi + \epsilon}{\theta}-1.$ Taking $\epsilon \to 0$ we get $\abs{\phi'(0)} \leq \frac{\pi}{\theta} -1.$ Substitute back into \eqref{eq: rowsum to derivative} gives the desired bound.
\end{proof}
\begin{remark}
%\cref{thm:sectorStableRowNormBound}'s bounds on $L_{\infty}$-norm and spectral norm are tight for $\inflMat$, but not for $\corMat.$
\cref{thm:sectorStableRowNormBound}'s bounds on  $\norm{\inflMat}_{\infty}$, $\norm{\inflMat_{\mu}}$, and $\norm{\corMat}_{\infty}$ are tight, even for homogeneous $\mu.$

For e.g., consider $f_{\mu} (z_1, \dots, z_{rk}) = \sum_{i=0}^{r-1} \prod_{j=ik+1}^{(i+1)k} z_j $ %z_1 \dots z_k + z_{k+1} \dots z_{2k} + \dots + z_{rk+1} \dots z_{(r+1)k}$, which is $S_{1/k}$-stable. 
For $r=2$, we have
$\inflMat_{\mu} = \begin{bmatrix} J_k & -J_k \\ -J_k & J_k \end{bmatrix} - I_{2k}$ and  $\norm{\inflMat}_{\infty}= \norm{\inflMat_{\mu}} = 2k-1.$
For arbitrary $r$ we get
 $\corMat_{\mu} = \begin{bmatrix} J_k & 0  & \dots & 0 \\ 0 & J_k & \cdots & 0 \\ \vdots  &  & &\vdots \\ 0 &  &\dots  & J_k \end{bmatrix}  - \frac{1}{r} J_{rk}$ with $J$ being the all ones matrix.
 
 \[\norm{\corMat}_{\infty} = k(1- \frac{1}{r}) + (r-1) k \frac{1}{r} = k (1 - \frac{2}{r}) \xrightarrow[r\to \infty]{} 2k .\]
 
 The bound on $\norm{\corMat_{\mu}}$ is tight in general, for e.g. consider $f(z_1, \dots, z_{2k}) = \epsilon z_1 \dots z_{2k} + (1-\epsilon)$ for small $\epsilon > 0,$ but is not tight for homogeneous distribution $\mu.$
 %Clearly, $\norm{\inflMat}_{\infty}$, $\norm{\inflMat_{\mu}}$, $\norm{\corMat}_{\infty}$ and $\norm{\corMat_{\mu}}$ are all $2k-1.$
\end{remark}

\begin{remark}\label{remark:generalization}
	The proof of \cref{thm:sectorStableRowNormBound} can be easily generalized to weaker types of stability. In particular, for the proof we only need to show that $\phi$ maps a ``large enough'' domain $A$ around $0$ to a ``bounded'' region $B$. One can then pre-compose $\phi$ with a map from the disk $D$ to $A$ and post-compose with a map from $B$ to the disk $D$, and apply Schwarz's lemma to the combination of these maps. We then derive a bound on $\abs{\phi'(0)}$ that only depends on the shape of regions $A$ and $B$ and what the derivative of the pre-composed and post-composed maps are at $0$.
	
	In the case of sector-stability $A$ was the strip $\Omega_{\alpha\pi/2}$ and $B$ was $\overline{\Omega_{\pi-\alpha\pi/2}}$. For weaker stability assumptions, one can get a smaller but large-enough region $A$, and a larger but small-enough region $B$. As an example, suppose that $\Gamma$ contains both $\R_{> 0}$ and a disk $D(1, \epsilon)$ around the point $1$, and $g_\mu$ is sector stable w.r.t.\ $\Gamma$. Assume further that $\Gamma$ is closed under inversion $z\mapsto z^{-1}$ (by choosing a potentially smaller $\epsilon$). By plugging in $z_i$ from $\R_{>0}$, we get that any positive linear of combination of $\bar{h}$ and $\bar{g}$ must be $D(1,\epsilon)$-stable. In particular, we still obtain that $\bar{g}(y)/\bar{h}(y)$ does not assume any value in $\R_{\leq 0}$ as long as $y\in D(1,\epsilon)$. This means that we can define a branch of $\log$ here that only takes values in $\Omega_\pi$. So our $A$ region will be the largest domain around $0$ with $\exp(A)\subseteq D(1,\epsilon)$ and $B$ will be $\Omega_{\pi}$. For any constant $\epsilon$, the derivative of a map $\phi$ from this $A$ to this $B$ will be $O(1)$.
\end{remark}

	\section{Sector-Stable Polynomials and Preserving Operations}

In this section, we show how certain natural operations affect the sector-stability of polynomials. In \cref{cor:hurwitzConstrained}, we show that the degree-$k$ part of a Hurwitz-stable (or $\Gamma_1$-stable) polynomial is $\Gamma_{1/2}$-stable. In \cref{thm:partition}, we show that given a homogeneous real-stable polynomial $g$, the sum of terms in $g$ whose $(T_1,\dots,T_k)$-degree is equal to $(c_1,\dots,c_k)$ is $\Gamma_{1/2^k}$-stable. 
These results are important ingredients in the proof of \cref{thm:monomer-dimer,cor:mixDerivative,thm:DPPsample}.

\begin{proposition}\label{prop: ssProperties}
The following operations preserve $\alpha$-sector-stability:
\begin{enumerate}
    \item\label{part: spec}Specialization: $g(z_1,\ldots,z_n)\mapsto g(a,z_2,\ldots,z_n)$, where $a\in \bar \Gamma_\alpha$.
    \item\label{part: scaling} Scaling: $g\mapsto g\star \lambda$, if $\lambda_i \in \R_{\geq 0} \forall i\in [n].$ 
    \item\label{part:dual} Dual: $g\mapsto g^*$, where $g(z)=\sum_{S\subseteq [n]} c_S z^S$ and  $g^* (z_1, \cdots, z_n) :=\sum_{S\subseteq [n]} c_S  z^{[n]\setminus S}$.
\end{enumerate}
\end{proposition}
\begin{proof}
Part \ref{part: spec} for $a\in \Gamma_\alpha$ holds by  the definition and for the closed boundary of $\Gamma_\alpha$ we can set $a$ to $0$ or $\infty$ by \cref{lemma:limitRoot}. Part \ref{part: scaling} holds  by the definition of sector-stability. For part \ref{part:dual}, \[g^{*} (z_1, \cdots, z_n)= z_1 \cdots z_n g(z_1^{-1}, \cdots, z_n^{-1}) \neq 0\]
for all $z_1, \cdots, z_n \in \Gamma_{\alpha}$, where we use $\Gamma_{\alpha}$-stability of $g$ and the fact that $z_1^{-1}, \cdots, z_n^{-1}$ are also in $\Gamma_{\alpha}$.
\end{proof}

\begin{lemma}[Homogenization] \label{lem:homogenize}
If multi-affine polynomial $g(z_1, \cdots, z_n) := \sum_{S\subseteq [n]} c_S z^S $ is $\Gamma_{\alpha}$-stable, then its \textit{homogenization} \[g^{\text{hom}} (z_1, \cdots, z_n, w_1, \cdots, w_n):= \sum_{S\subseteq [n]} c_S z^S w^{[n]\setminus S}\] is multi-affine, homogeneous of degree $n$, and $\Gamma_{\alpha/2}$-stable.
\end{lemma}
\begin{proof}
One can rewrite $g^{\text{hom}}$ as
\[g^{\text{hom}} (z_1, \cdots, z_n, w_1, \cdots, w_n) = w_1 \cdots w_n g(\frac{z_1}{w_1}, \cdots, \frac{z_n}{w_n}).\]
For any $z_1, \cdots, z_n, w_1 \cdots, w_n \in \Gamma_{\alpha/2}$, we have $\frac{z_i}{w_i} \in \Gamma_{\alpha} \forall i\in [n],$ thus the RHS is nonzero by $\Gamma_{\alpha}$-stability of $g.$
\end{proof}
\begin{lemma} \label{cor:homogenize monomer-dimer poly}
Consider graph $G = G(V,E)$ on $n$ vertices with edge weight $w : E \to \R_{\geq 0}$ and vertex weight $\lambda: V \to \R_{\geq 0}$. For $S \subseteq V$, let $m_S := \sum_M \weight(M) =  \sum_M(\prod_{e\in M} w(e) \prod_{v\not\in S} \lambda(v))$ where the sum is taken over all perfect matching $M$ of $S$. 
	The following polynomial is $\Gamma_{1/2}$ stable
	\[f(z_1, \cdots, z_n, y_1, \cdots, y_n) = \sum_{S \subseteq V} y^S z^{[n] \setminus S} m_S. \]
\end{lemma}
% \begin{lemma}[Complement/Dual] \label{lem:dualPoly}
% If multi-affine polynomial $g(z_1, \cdots, z_n) := \sum_{S\subseteq [n]} c_S z^S $ is $\Gamma_{\alpha}$-stable, then its \textit{complement} or \textit{dual}
% \[g^* (z_1, \cdots, z_n) :=\sum_{S\subseteq [n]} c_S  z^{[n]\setminus S}\] is multi-affine and $\Gamma_{\alpha}$-stable.
% \end{lemma}
% \begin{proof}
% \[g^{*} (z_1, \cdots, z_n)= z_1 \cdots z_n g(z_1^{-1}, \cdots, z_n^{-1}) \neq 0\]
% for all $z_1, \cdots, z_n \in \Gamma_{\alpha}$, where we use $\Gamma_{\alpha}$-stability of $g$ and the fact that $z_1^{-1}, \cdots, z_n^{-1}$ are also in $\Gamma_{\alpha}$
% \end{proof}
% We define some notations that will be useful for the proof of \cref{thm:hurwitzPartition}, \cref{thm:partition}. 
The class of sector-stable polynomials was studied in \cite{sendov_sendov_2019}, where the authors proved that symmetrization preserves sector-stablity of univariate polynomials with nonnegative coefficients.
Given a univariate complex polynomial $p(z)=a_nz^n+\ldots+a_1z+a_0$, its symmetrization with $n$ variables is defined as
\[P(z_1,\ldots, z_n)=\sum_{k=0}^n\frac{a_k}{{n\choose k}}S_k(z_1,\ldots,z_n),\]
where $S_k(z_1,\ldots,z_n)=\sum_{1\leq i_1<\ldots<i_k\leq n}z_{i_1}\ldots z_{i_k}$. By the definition, $P(z,\ldots,z)=p(z)$. We call $(z_1,\ldots,z_n)$ a solution of $p$, if $P(z_1,\ldots,z_n)=0$. 
Define a closed set $\Omega\subseteq \mathbb{C}^*$ the locus holder of $p$, if every solution of $p$ has a point in $\Omega.$ Call a minimal by inclusion locus holder $\Omega$ a locus of $p$. For examples and properties of locus holders see \cite{Sendov2014LociOC}. Note that any polynomial is stable with respect to the complement of its locus. The next result shows that symmetrization of a univariate sector-stable polynomial with non-negative is sector-stable. Note that this result is not true if we drop the assumption of nonnegative coefficients.
\begin{proposition}[Theorem 1.1 \cite{sendov_sendov_2019}]
Let $p(z)$ be a univariate $\Gamma_{\alpha}$-sector-stable polynomial with nonnegative coefficients. Then $\Gamma_\alpha$ is the locus holder of $p(z)$.
\end{proposition}
% \begin{corollary}
% Let $p(z)$ be a uni-variate $\alpha$ sector-stable polynomial of degree $n$ with nonnegative coefficients. Then its symmetrization $P(z)$ with $n$ variable is also $\alpha$ sector-stable.
% \end{corollary}

For the following results, we consider degree of a polynomial $g$ with respect to indices in a given set $S$.  When the set $S$ and $g$ is specified let $k_{\max}, k_{\min}$ be the maximum and minimum $S$-degree among monomials in $g$.

\begin{lemma} \label{lem:kmaxKmin}
Let $U := \prod_i \Gamma_{\alpha_i}\subseteq \C,$ $S \subseteq [n].$ 
	If $g\in \C[z_1,\dots,z_n]$ is  $U$-stable, then $g_{k_{\max}^S}^S, g_{k_{\min}^S}^S$ are also $U$-stable.
\end{lemma}
\begin{proof}
We may re-index $z_i$ so that $S = [t]$ for some $t \leq n.$ W.l.o.g., assume that this is already done.

For simplicity of notation, below we omit the superscript $S.$
% Let $k_{\max}: = k_{\max}^S, k_{\min}:= k_{\min}^S.$
Observe that $U$ is open and $g_{k_{\max}}, g_{k_{\min}}$ are not identically zero, by definition.

For $\lambda \in \R_{>0}$ let \[g^{\lambda}(z_1, \cdots, z_n) := \frac{1}{\lambda^{k_{\max}}} g(\lambda z_1, \cdots,\lambda z_t, z_{t+1}, \cdots, z_n) = g_{k_{\max}}(z_1, \cdots, z_n) + \sum_{k=0}^{k_{\max}-1} \frac{g_k (z_1, \cdots, z_n)}{\lambda^{k_{\max}-k}}\]
	
	Clearly, $g^{\lambda}$ is $U$-stable, and $\lim_{\lambda \to \infty} g^{\lambda} = g_{k_{\max}} $, so by Lemma \ref{lemma:limitRoot}, $g_{k_{\max}}$ is $U$-stable.
	Similarly, $g_{k_{\min}} = \lim_{\lambda \to 0^+} \frac{1}{\lambda^{k_{\min}}}g(\lambda z_1, \cdots,\lambda z_n) $ is $U$-stable.
\end{proof}
As a consequence, we can prove partial derivatives preserve sector stability.
	\begin{corollary}
If $p(z_1,\ldots z_n)$ is a multiaffine polynomial, then the partial derivative of $p$ with respect to any variable $z_i$ in $i \in [n]$, which we denote by $\partial_i p$, is sector stable.
\end{corollary}
\begin{remark}
In general taking derivatives of non-multiaffine polynomials does not preserve sector stability. For example, let $x, y, z_1, \dots , z_n$ be variables. Look at the polynomial $p = (xz_1+yz_2)(xz_2+yz_3)(xz_3+yz_4) \dots (xz_n+yz_1)$. This is $\Gamma_{1/2}$-sector stable. Now differentiate w.r.t. each $z_i$ once, and then set each $z_i$ it to zero. What
 you end up with is $x^n+y^n$. This is only  $\Gamma_{1/n}$-sector-stable.
\end{remark}
% 	Theorem \ref{thm:hurwitzPartition} => Theorem 1.
	\begin{theorem}[Hurwitz-stable intersected with one partition constraint]
		\label{thm:hurwitzPartition}
		Suppose $g(z_1,\dots,z_n)$ is a $\Gamma_1$-stable polynomial with constant parity (the degree of every monomial is even or odd). Then $g_k$ is $\Gamma_{1/2}$-stable or identically 0.
	
	More precisely, for $k\in [k_{\min}, k_{\max}] $ with $k\equiv k_{\max} \pmod{2}$,  $g_{k}$ is $\Gamma_{1/2}$ stable.
	\end{theorem}

	\begin{proof}
	
	 \cref{lem:kmaxKmin} with $S = [n]$ and $U = \Gamma_1^n$ implies $g_{k_{\max}}, g_{k_{\min}}$ are $\Gamma_1$-stable.  W.l.o.g., we assume $k_{\max} > k_{\min}\geq 0$, otherwise there is nothing to prove.
	 
		Fix arbitrary $z_1,\dots,z_n \in \Gamma_{1/2}.$ Let $h(z)=\frac{1}{z^{k_{\min}}} g(z_1 z,z_2 z,\dots,z_n z)$. Note that $h(0) = g_{k_{\min}} (z_1, \cdots, z_n) \neq 0$ by $\Gamma_1$-stability of $g_{k_{\min}}.$ Note also that all terms in $h$ has even degree in $z$, and the highest degree term is $g_{k_{\max}} (z_1, \cdots, z_n) z^{k_{\max}-k_{\min}}$ with $g_{k_{\max}} (z_1, \cdots, z_n) \neq 0$ by $\Gamma_1$-stability of $g_{k_{\max}}$. By substituting $z = y^{1/2}$ in $h$, we obtain a polynomial $\tilde{h}(y): = h(y^{1/2}) $ that satisfies $\tilde{h}(y)\neq 0$ whenever $y \in \bar{S}_1 \cup \set*{0}$. Indeed, $\tilde{h}(0)= h(0) \neq 0.$ For $y \in \bar{S}_1$, we have $z = y^{1/2} \in \bar{S}_{1/2}$ thus $(z_i z)_{i=1}^n \in \Gamma_1^n$, and $\tilde{h}(y) = g(z_1 z, \cdots, z_n z) \neq 0$ by $\Gamma_1$-stability of $g.$
		
		Let $\lambda_1, \cdots, \lambda_{d}$ be the roots of $\tilde{h}(y)$ where $d:=\deg(\tilde{h}) = \frac{k_{\max} - k_{\min}}{2}.$ As argued earlier, $\lambda_i \in (\C \setminus (\bar{S}_1 \cup \set*{0})) = H_{-\pi/2} .$ Fix $k \in [k_{\min}, k_{\max})$ with $k \equiv k_{\max} \pmod{2}.$ By half-plane stability of symmetric polynomial ( \cref{thm:symmetric}), \[g_k(z_1, \cdots, z_n) = g_{k_{\max}}(z_1, \cdots, z_n) e_{t} (\lambda_1, \cdots, \lambda_d) \neq 0\] where $t: = \frac{k_{\max} - k}{2} \in \N$. 
		
		%Note that $h$ has either odd or even degree terms, and in particular
		%\[ h(z)=\pm h(-z). \]
		%So the roots of $h$ are nonzero, and come in pairs that are negations of each other. %(except for maybe one root at $0$). 
		
% 		Let $\pm \lambda_1,\dots,\pm \lambda_d$ be these roots (there might be an additional zero root). Our goal is to show that $g_k(z_1,\dots,z_n)\neq 0$. Note that if $k\in \set{k_{\max}, k_{\min}}$ then this follows from $\Gamma_1$-stability of $g_{k_{\max}}, g_{k_{\min}}$. For $k\in (k_{\min}, k_{\max})$ and $k\equiv k_{\max}\pmod{2}$, we have \[g_k(z_1,\dots,z_n) = g_{k_{\max}} (z_1, \cdots, z_n) e_{t} (\lambda_1,-\lambda_1,\lambda_2,-\lambda_2,\dots,\lambda_d,-\lambda_d )\]
% 		where $t=k_{\max}-k$ is an even number, and $\pm \lambda_i$ are the nonzero roots of $h.$ 
		
% 		But this is equal to either the $k$-th or the $(k-1)$-st (based on the parity) elementary symmetric polynomial of $\pm \lambda_1,\dots,\pm \lambda_d$. 
% 		So it is enough to show that for every even number $t$, we have
% 		\[ e_t(\lambda_1,-\lambda_1,\lambda_2,-\lambda_2,\dots,\lambda_d,-\lambda_d)\neq 0. \]
% 		But it is an elementary observation that the above is equal to
% 		\[ e_{t/2}(-\lambda_1^2,-\lambda_2^2,\dots,-\lambda_d^2). \]
% 		Now note that $\pm \lambda_i$ is a root of $h$, so it cannot be in $\set{z\given \abs{\arg(z)}<\pi/4}$. So for every $i$, $\abs{\arg(\lambda_i)}\in (\pi/4, 3\pi/4)$. This means that $\abs{\arg(-\lambda_i^2)}<\pi/2$. From Hurwitz-stability of the elementary symmetric polynomials (Theorem \ref{thm:symmetric}), it follows that the above is nonzero.
	\end{proof}
The next corollary results in DPP sampling on $P_0$ matrix $A\in \R^{n\times n}$, where $A + A^T$ is PSD, and the sampling from monomer-dimer of fixed size.
	\begin{corollary} \label{cor:hurwitzConstrained}
	Suppose $g(z_1,\dots,z_n) \in \R[z_1, \cdots, z_n]$ is $\Gamma_1$-stable, then $g_k$ is either identically 0 or $\Gamma_{1/2}$-stable.
	\end{corollary}
	\begin{proof}
	Define the even and odd parts $g_e$ and $g_o$ of $g$ as in Theorem \ref{thm:hurwitzSameParityPart}. If $g_e\equiv 0$ or $g_o \equiv 0$ then the claim follows from Theorem \ref{thm:hurwitzPartition}.
	
	Suppose $g_e, g_o \not\equiv 0$, then they are $\Gamma_1$-stable by Theorem \ref{thm:hurwitzSameParityPart}. The claim follows by applying Theorem \ref{thm:hurwitzPartition} to $g_e$ ($g_o$) if $k$ is even (odd resp.)
	\end{proof}
	\cref{lemma:p0MatrixHurwitz} and \cref{cor:hurwitzConstrained} together imply the following corollaries.
 	\begin{corollary} \label{cor:p0Constrained}
	Consider $A\in \R^{n \times n}$ where $A + A^T$ is positive semi-definite, then \[f_k(z_1, \cdots, z_n) = \sum_{S \in \binom{[n]}{k}} \det(A_{S,S}) z^{[n]\setminus S} \] and its dual \[f_k^*(z_1, \cdots, z_n) = \sum_{S \in \binom{[n]}{k}} \det(A_{S,S}) z^{S}  \] are either identically $0$ or $\Gamma_{1/2}$-stable, and has nonnegative real coefficients.
	\end{corollary}
	\begin{corollary}\label{cor:monomerDimerConstrained}
	Consider a graph $G = G(V,E)$ on $n$ vertices. For $S \subseteq V$, let $m_S$ be the number of perfect matching on $S$. Then \[f_k(z_1, \cdots, z_n) = \sum_{S \in \binom{[n]}{k}} m_S z^{[n]\setminus S}\] and its dual $f_k^*$ are either identically $0$ or $\Gamma_{1/2}$-stable.
	\end{corollary}
% 	Consequences: The degree $k$-part of monomer-dimer polynomial is $\Gamma_{1/2}$-stable or $\equiv 0$, and the degree $k$-part of $\det(\diag(z_1, \cdots, z_n) + L )$ 
	\begin{lemma}\label{lem:hom}
		Suppose that $p(x, y)$ is a homogeneous polynomial with coefficients in $\C$, defined as
		\[ p(x,y)=\sum_{i}c_ix^iy^{d-i}. \]
		If $p$ is $(\overline{\Gamma_\alpha}\times \overline{\Gamma_\beta})$-stable for $\alpha+\beta\geq 1$, then the sequence of $c_i$ will have no holes (zeros in between nonzeros).
	\end{lemma}
	\begin{proof}
		We may as well assume that $c_0, c_d\neq 0$, otherwise we can factor out extra powers of $x$ and $y$. Let $g(z)=p(z, 1)$.	Then $g$ is $\overline{\Gamma_1}$-stable. This is because every $z\in \overline{\Gamma_1}$ can be written as $x/y$ for $x\in \overline{\Gamma_\alpha}$ and $y\in \overline{\Gamma_\beta}$. So $g(z)=p(x/y,1)=p(x, y)/y^d\neq 0$.
		Since $g$ is $\overline{\Gamma_1}$-stable and has no zero root, its roots must be in the left half-plane $\set{z\given \Re(z)<0}$. But then $c_{d-i}/c_d$ is going to be up to a plus/minus sign the $i$-th elementary symmetric polynomial of the roots of $g$. Since elementary symmetric polynomials are half-plane-stable for every open half plane, all the coefficients of $g$ must be nonzero.
	\end{proof}

    %Theorem \ref{thm:partition} => Theorem  \ref{thm:mixDerivative}
	\begin{theorem}\label{thm:partition}
		Suppose that $g(z_1,\dots,z_n)$ is a homogeneous $\Gamma_1$-stable polynomial. Let $T_1, \dots, T_k$ be a partition of $[n]$ and $(c_1,\dots,c_k)\in\Z_{\geq 0}^n$. Define the $(T_1,\dots,T_k)$-degree of a monomial $z_1^{t_1}\cdots z_n^{t_n}$ as $(\sum_{i\in T_1} t_i, \sum_{i\in T_2} t_i, \dots, \sum_{i\in T_k} t_i)$. Let $h$ be the sum of terms in $g$ whose $(T_1,\dots,T_k)$-degree is equal to $(c_1,\dots,c_k)$. Then $h$ is either identically zero, or is $\Gamma_{1/2^{k}}$-stable.
	\end{theorem}
	\begin{proof}
		Let $h_i$ be the polynomial obtained from $g$ by retaining the terms whose $(T_1,\dots,T_i)$-degree is $(c_1,\dots,c_i)$. Then $h_0=g$ and $h_k=h$. If $h_i \equiv 0$ for some $i$, then $h_k \equiv 0$. W.l.o.g., we assume $h_i\not \equiv 0 \forall i.$   We will inductively prove that $h_i$ is $ \parens*{\prod_{j\in T_1\cup\cdots\cup T_i} \Gamma_{\alpha}\times \prod_{j\in T_{i+1}\cup\cdots\cup T_k} \Gamma_{\beta_i}}$-stable for $\alpha=1/2^k$ and $\beta_i=1-(2^i-1)/2^k$. Let $\prod_i := \parens*{\prod_{j\in T_1\cup\cdots\cup T_i} \Gamma_{\alpha}\times \prod_{j\in T_{i+1}\cup\cdots\cup T_k} \Gamma_{\beta_i}}.$ Note that $\prod_{i+1} \subseteq \prod_i \forall i.$
		
		Note that $\beta_0=1$, and by assumption $g=h_0$ is $\Gamma_1$-stable. So it is enough to prove the induction step. Assume the statement is true for $h_i$ and let us prove it for $h_{i+1}$. Fix $(z_1,\dots,z_n) \in \prod_{i+1}.$ We will show $h_{i+1}(z_1,\dots,z_n)\neq 0$. Note that we can get $h_{i+1}$ from $h_i$ by retaining the terms whose $T_{i+1}$- degree is $c_{i+1}$. Take two variables $x$ and $y$, and look at the polynomial $p(x, y)=h_i(u_1,\dots,u_n)$, where
		\[
			u_j:=\begin{cases}
				z_j & \text{if }j\in T_1\cup\cdots\cup T_i,\\
				xz_j & \text{if }j\in T_{i+1},\\
				yz_j & \text{if }j\in T_{i+2}\cup\cdots\cup T_k.\\
			\end{cases}
		\]
		
		Note that $p$ is a homogeneous polynomial (of some degree $d$). This is because $h_i$ is homogeneous in variables from $T_{i+1}\cup\cdots \cup T_k$. %Let $\tilde{h}: = h_i.$ 
		Let $ c_{\max}, c_{\min}$ be the maximum and minimum $T_{i+1}$-degree in $h_i$ respectively. Note that the coefficient of $x^{c}y^{d-c}$ in $p(x, y)$ is exactly $h_{i,c}^{T_{i+1}}(z_1, \cdots, z_n)$, where $h_{i,c}^{T_{i+1}}$ are sum of terms in $h_i$ whose $T_{i+1}$-degree is $c.$
		We will show that  
		the coefficient of $x^{c}y^{d-c}$ in $p(x, y)$ are nonzero, where $c\in [c_{\min}, c_{\max}]$.  This immediately implies stability of $h_{i+1}$, as $c_{i+1} \in [c_{\min}, c_{\max}] $ since $h_{i+1} \not\equiv 0$. For $c\in \set*{c_{\min}, c_{\max}},$ $h_{i,c}^{T_{i+1}}$ is $\prod_i$-stable by inductive assumption on $h_i$ and \cref{lem:kmaxKmin}, thus $h_{i,c}^{T_{i+1}}(z_1, \cdots, z_n) \neq 0$, as $(z_1, \cdots, z_n) \in \prod_{i+1} \subseteq \prod_i.$
		
% 		Note that $h_{i+1}(z_1,\dots,z_n)$ is the coefficient of $x^{c_{i+1}}y^{d-c_{i+1}}$ term in $p(x, y)$. So to prove the stability of $h_{i+1}$ it is enough to show that the coefficients of $p$ are nonzero, and to do that we will use \cref{lem:hom}. 
		For the remaining $c\in (c_{\min}, c_{\max})$ we will use \cref{lem:hom}.
		Let $x\in \overline{\Gamma_{\beta_i-\alpha}}$ and $y\in \overline{\Gamma_{\beta_i-\beta_{i+1}}}$. These choices make sure that $xz_j\in \Gamma_{\beta_i}$ and $yz_j\in \Gamma_{\beta_i}$ for appropriate indices $j$. By the inductive assumption, we have $p(x, y)\neq 0$. So $p$ is $(\overline{\Gamma_{\beta_i-\alpha}}\times \overline{\Gamma_{\beta_{i}-\beta_{i+1}}})$-stable. If this stability satisfies the assumptions of \cref{lem:hom}, we are done. So it is enough to check
		\[ (\beta_i-\alpha)+(\beta_{i}-\beta_{i+1})\geq 1. \]
		An easy calculation shows that the l.h.s.\ is
		\[ 2\beta_i-\alpha-\beta_{i+1}=2-\frac{2^{i+1}-2}{2^k}-\frac{1}{2^k}-1+\frac{2^{i+1}-1}{2^k}=1. \]
	\end{proof}
	\begin{conjecture}
		With the same assumptions as in \cref{thm:partition}, we have $h$ is either identically zero or $\Gamma_{1/k}$-stable.
	\end{conjecture}

For any distribution $\mu$, define its Newton polytope, $\newt(\mu)$ as the convex hull of its support,
\[\newt(\mu):=\conv(\{S:\mu(S)>0\}).\]
Next, we show that the $\l_1$ edge lengths of the Newton polytope of a $\Gamma_{1/k}$-sector-stable distribution are $O(k)$.

\begin{lemma} \label{lem:edgeBound}
Let $\mu:2^{[n]}\rightarrow\mathbb{R}$ be a $\Gamma_{1/k}$-sector-stable distribution, then the length of edges of $\newt(\mu)$ is at most $2k$.
\end{lemma}
\begin{proof}
First, we show that for any face $F$ of $\newt(\mu)$ there exists a $\Gamma_{1/k}$-sector-stable polynomial with support equal to the face $F$. Since $F$ is a face of $\newt(\mu)$ there exists some vector $w=(w_1,\ldots,w_n)$ that 
\[F=\arg\max\{\langle w,x\rangle|x\in \newt(\mu)\}.\]
Let $g_\mu$ be the generation polynomial of $\mu$, then
\[g_\mu(t^{w_1}z_1,\ldots,t^{w_n}z_n)=\sum_{\alpha \in\mathbb{Z}^n_{\geq 0}}\text{coeff}_{g(z^\alpha)}t^{\langle w,\alpha \rangle}z^\alpha.\]
Now, if we take the limit $t\rightarrow \infty$ only the coefficients of the terms that are the supports of $F$ remain, i.e.,
\[\frac{1}{t^{\max\langle w,x \rangle }}g(t^{w_1}z_1,\ldots,t^{w_n}z_n)=\sum_{\alpha\in F}\text{coeff}_{g(z^\alpha)}z^\alpha+\sum_{\alpha\not\in F}t^{-\delta\alpha}\text{coeff}_{g(z^\alpha)}z^{\alpha}\]
\[g_F=\lim_{t\rightarrow\infty}\frac{1}{t^{\max\langle w,x\rangle }}g(t^{w_1}z_1,\ldots,t^{w_n}z_n)=\sum_{\alpha\in F}\text{coeff}_{g(z^\alpha)}z^{\alpha}.\]
Note that linear scaling of variables and taking limit with respect to $t$, preserves sector-stability of the function, therefore, $g_F$ is a sector-stable polynomial. By applying the same argument again, we can constraint the Newton polytope to lower dimensional faces and yet still preserves sector-stability, until we get an edge. As a result, the corresponding polynomial for each edge should be also $\Gamma_{1/k}$-sector-stable. Now, assume the the contrary that there exists an edge $(\alpha,\alpha')$ of $\newt(\mu)$ with length more than $2k$. Then we should have that 
\[g_{(\alpha,\alpha')}(z)=az^{\alpha}+bz^{\alpha'} = b z^{\alpha} (a b^{-1} + z^{\alpha'-\alpha}) \]
is $\Gamma_{1/k}$-sector-stable, where $a=\text{coeff}_{g(z^\alpha)}$ and
$b=\text{coeff}_{g(z^{\alpha'})}$ are nonzero. 

If $\abs{\alpha-\alpha'}_1>2k$, we can set $z_i \in \Gamma_{1/k}$ for $i \in \alpha \Delta \alpha'$ so that $z^{\alpha'-\alpha}$ can take any values in $\C.$ Therefore, $g_{(\alpha,\alpha')}(z)$ is not $\Gamma_{1/k}$-stable, a contradiction. 
% Now, let
%  $z=(z_1,\ldots,z_n)$, where $z_i=1$ for $i\in\alpha\cap\alpha'$, $z_j=e^{i\frac{\pi}{2k}}$ for $j\in\alpha\setminus\alpha'$ and $z_l=e^{-i\frac{\pi}{2k}}$ for  $l\in\alpha'\setminus\alpha$. Then $g_{(\alpha,\alpha')}(z)$ has a root in $\frac{\pi}{2k}$ sector when $|\alpha-\alpha'|>2k$, a contradiction.
% \june{this is problematic, $e^{i\pi/2k} \not \in \Gamma_{1/k}.$ Also, you need to rescale them to deal with the coefficients $a, b.$}
\end{proof}

	\section{Fractionally Log-Concave Polynomials}\label{sec:log-concave}

In this section, first, we show that any sector-stable polynomial is a fractionally log-concave polynomial as well. Then by analyzing properties of fractionally log-concave polynomials we show that entropy of marginals gives a constant approximation for the entropy of fractionally log-concave distributions. This leads to a multiplicative-approximation on the \textit{logarithm} of the size of the support of a sector-stable polynomial (see \cref{lem:approximateLogSupport}). Our techniques are a natural generalization of the results obtained by \textcite{AOV18}. See also \cite{ES20} for recent alternative techniques for proving similar entropy-based inequalities. An immediate consequence of our results is a multiplicative-approximation for the \textit{logarithm} of the size of the support of the monomer-dimer model (\cref{cor:2kmatching}).  
		\begin{lemma}\label{lemma:sectorStableToFracLC}
		For $\alpha \in [0,1/2]$, if polynomial $f\in \R_{\geq 0}[z_1, \cdots, z_n]$ is $\Gamma_{2\alpha}$-sector-stable then $f$ is $\alpha$ fractionally log-concave.
		\end{lemma}
		\begin{proof}
		Let $\mu: 2^{[n]} \to  \R_{\geq 0}$ be the distribution generated by $f.$
		
		First we claim that, it is enough to prove fractional log-concavity at $\vec{1}.$ For an arbitrary vector $\vec{v} \in \R_{> 0}$ let $f^{v} (z_i) = f(\set{v_i^{\alpha} z_i}).$ Note that $f^v$ is sector stable, and 
		$$\nabla^2 \log f(\set{z_i^{\alpha}}) |_{\vec{v}} = D_{\vec v}\big( \nabla^2 \log  f^v(\set{z_i^{\alpha}}) |_{\vec{1}}\big) D_{\vec v}, $$
		where  $D_{\vec v}=\diag{\set{v_i^{-1}}}$. So, we proceed by replacing $f$ with $f^v.$
		
Let $H = \nabla^2 \log f(\set{z_i^{\alpha}}) |_{\vec{1}} $ then
\[H_{ij} = \begin{cases} \alpha(\alpha-1) \P{i} - \alpha^2 \P{i}^2 &\text{ if $j=i$}\\
\alpha^2 (\P{i \land j} - \P{i} \P{j}) &\text{ otherwise}\end{cases}.\]
%Let $\Psi'_{ij} =\P{j\given i} -\P{j}$ for $j\neq i$ and $\Psi'_{ii} = 1-\P{i}.$ 
Since the row norm of $\corMat_{\mu}$ is bounded by $\frac{2}{2\alpha}$, its maximum eigenvalue $\lambda_{\max} (\corMat_{\mu})$ is at most $\frac{1}{\alpha}.$ Therefore,
\[\frac{1}{\alpha} I \geq \corMat_{\mu} =\frac{1}{\alpha^2} \diag{(1/\P{i})_i} H + \frac{1}{\alpha} I,  \]
hence, $H 
\leq 0$ and $\log f(\set{z_i^\alpha})$ is concave.

% We can get $1/\alpha-1$ (correct bound for $\alpha=1$) if we define $\tilde{P}_{ii} = -\P{i}$ though.   
\end{proof}
\begin{remark}
    Observe that the proof of \cref{lemma:sectorStableToFracLC} implies $\lambda_{\max} (\corMat_{\mu}) \leq \frac{1}{\alpha} $ is equivalent to $\alpha$-fractional log-concavity of $\mu.$ 
\end{remark}
\cref{lem:hom,lemma:sectorStableToFracLC} imply that for $\Gamma_{\alpha}$-sector stable $\mu: 2^{[n]} \to \R_{\geq 0}$, the homogenization $\mu^{\hom}: \binom{[2n]}{n} \to \R_{\geq 0}$ of $\mu$ is $\Gamma_{\alpha/2}$-sector stable and $\alpha/4$-fractionally log concave. In \cref{lem: homogenization tighter bound}, we prove the stronger statement that $\mu^{\hom}$ is $\alpha/2$-fractionally log concave.
 %We can prove a better bound for $\lambda_{\max} (\corMat_{\mu}) $ the homogenization of a sector-stable distribution $\mu$.
 \begin{lemma} \label{lem: homogenization tighter bound}
 Consider distribution $\mu: 2^{[n]} \to \R_{\geq 0}$
 that is generated by a $\Gamma_{\alpha}$-sector-stable polynomial $f.$ Let $\nu: = \mu^{\hom}$ be the homogenization of $\mu.$ We have $\lambda_{\max}(\corMat_{\nu}) \leq \frac{2}{\alpha} $, or, equivalently, the homogenization $f^{\hom}$ of $f$ is $\alpha/2$-fractionally log-concave.
 \end{lemma}
 \begin{proof}
 Let $\Omega = \set*{1, \dots, n}$ and $\bar{\Omega} = \set*{\bar{1}, \dots, \bar{n}}.$ For set $S \subseteq \Omega$, let $\overline{S}\subseteq \overline{\Omega}: = \set*{\bar{i} \given i \in S}.$ %For e.g., if $S =[n-1]$ then $\bar{S} = \set{\bar{n}}.$ 
 Recall that 
 \[\nu(U) = \begin{cases} \mu( U \cap \Omega) &\text{if $U =  S \cup \overline{(\Omega \setminus S)}$} \\ 0 &\text{else} \end{cases} \]
 %Recall $\nu$ is distributed over $\set*{1, \dots, n, \bar{1}, \dots, \bar{n}}$

 We let $\P{i}: =\P_{U \sim \nu} {i\in U} = \P_{S \sim \mu} {i\in S}$ and $\P{\bar{i}} :=\P_{U \sim \nu} {\bar{i}\in U} = \P_{S \sim \mu} {i\not\in S}.$
 
 Note that $\corMat_{\nu}(i, i) = -\corMat_{\nu}(i, \bar{i}) = \P{\bar{i}}$ and $\corMat_{\nu} (\bar{i}, \bar{i})=  - \corMat_{\nu}(\bar{i}, i) = \P{i}.$ For $i\neq j$, we can write
\begin{align*}
    \corMat_{\nu}(i, j) &= - \corMat_{\nu}(i, \bar{j}) = \P{\bar{i}} \inflMat_{\mu}(i,j)\\
    \corMat_{\nu}(\bar{i}, \bar{j})&= - \corMat_{\nu}(\bar{i}, j)   = \P{i} \inflMat_{\mu}(i,j)\\
\end{align*}
Let $D := \diag{(\P{i})_{i=1}^n}$ and $\overline{D} := \diag{(\P{\bar{i}})_{i=1}^n}.$ We can rewrite $\corMat_{\nu}$ as a block matrix in term of matrix $A : =\inflMat_{\mu} + I$ as follow
\[ \corMat_{\nu} = \begin{bmatrix} \overline{D} A & - \overline{D}A \\ - D   A & DA \end{bmatrix}\]
We consider the left eigenvectors of $\corMat_{\nu}.$ %Let $\lambda_1(A) \geq \dots\geq \lambda_n(A)$ be the eigenvalues of $A$, with corresponding
Recall that all eigenvalues of $\inflMat_{\mu}$ are real.%, and $A = \inflMat_{\mu} + I$ has the same 
Let $ v_1, \cdots, v_n \in \R^n$ be a basis of left eigenvectors of $\inflMat_{\mu}$, with corresponding eigenvalues $\lambda_1(\inflMat_{\mu}) \geq \dots\geq \lambda_n(\inflMat_{\mu}).$ For $i\in [n],$ let $w_i \in \R^{2n}$ be the concatenation of $v_i$ and $-v_i$ i.e. $w_i^t : = \begin{bmatrix} v_i^t & -v_i^t \end{bmatrix}.$ Then $\set*{w_i}$ are linearly independent, and are left eigenvector of $\corMat_{\nu} $ with eigenvalues $\set*{\lambda_i+1}$, since
\[w_i^t \corMat_{\nu} = \begin{bmatrix} v_i^t & -v_i^t \end{bmatrix} \begin{bmatrix} \overline{D} A & - \overline{D}A \\ - D   A & DA \end{bmatrix} =  \begin{bmatrix} v_i^t  (\overline{D} + D) A & -v_i^t (\overline{D} + D)A \end{bmatrix} =\begin{bmatrix} v_i^t  A & -v_i^t A \end{bmatrix} = (\lambda_i+1) w_i^t\]
On the other hand, for $i\in [n]$, consider the vector $u_i\in \R^{2n}$ defined by $u_i^t: = \begin{bmatrix} \textbf{e}_i^t D  & \textbf{e}_i^t \overline{D}\end{bmatrix}  $ where $\textbf{e}_i$ is the $i$-th standard basis vector of $\R^n.$ Observe that $u_i \neq 0$ since either $\P{i}$ or $\P{\bar{i}}$ must be nonzero, and $u_i^t \corMat_{\nu} = 0 .$ Moreover, $\set*{u_i}$ are linearly independent.

Now, let $W , U$ be the $n$-dimensional subspaces of $\R^{2n}$ spanned of $\set*{w_i}$ and $\set*{u_i}$ respectively. We show that $W \cap U = \set*{\textbf{0}},$ then conclude that the vectors $\set*{u_i}\cup \set*{w_i}$ are linearly independent, and form a basis of (left) eigenvectors of $\corMat_{\nu}.$
Hence, the spectrum of $\corMat_{\nu}$ is the union of $\set*{\lambda_i+1}_{i=1}^n$ and $n$ copies of $0.$ In particular, $\lambda_{\max} (\corMat_{\nu})\leq \lambda_1(\inflMat_{\mu}) + 1 = \frac{2}{\alpha} .$

Indeed, suppose $w\in W \cap U.$ We can write $w^t = \begin{bmatrix} y^t & -y^t \end{bmatrix}$ for some $y \in \R^n$ and $w^t = \begin{bmatrix} x^t D & x^t \overline{D} \end{bmatrix}$ for some $x \in \R^n.$Then
\[0=y(i) - y(i) = w(i) + w(i+n)= x(i) \P{i} + x(i) \P{\bar{i}} = x(i)\]
where we use $y(i)$ ($x(i), z(i)$ resp.) to denote the $i$-th entries of vector $y.$ Now, all entries of $x$ are 0, so $w = \textbf{0}.$ 
 \end{proof}
 \begin{remark}
 \cref{lem: homogenization tighter bound} is tight. For example, take $ f_{\mu}(x_1, x_2) = x_1 + x_2$, which is $\Gamma_1$-stable, then $\lambda_{\max} (\corMat_{\mu^{\hom}}) = 2.$
 \end{remark}
 For any probability distribution $\mu$ over a finite set $\Omega$, define its entropy as $\mathcal{H}(\mu)=\sum_{\omega\in\Omega}\mu(\omega)\log\frac{1}{\mu(\omega)}$. Recall the marginal probability of element $i\in\Omega$, $\mu(i)$,  is the probability that $i$ is in a sample from $\mu$, i.e., $\mu(i)=\P_{S\sim \mu}{i\in S}.$ For any probability distribution, by sub-additivity of entropy, we know that entropy of marginals is an upper bound on the entropy,
 \[\sum_i\mathcal{H}(\mu(i))\geq \mathcal{H}(\mu).\]

 The next lemma, which is an analogous of Theorem 5.2. in \cite{AOV18}, leads to a lower bound for the entropy of fractionally log-concave polynomials.

	\begin{lemma}\label{lemma: entropy lower bound}
	For any $\alpha$ fractionally log-concave distribution $\mu:2^{[n]}\rightarrow \mathbb{R}$ with marginal probabilities $\mu_1,\ldots,\mu_n$, we have
	\[\mathcal{H}(\mu)\geq \alpha\sum_i\mu(i)\log(\frac{1}{\mu(i)}).\]
	\end{lemma}
	\begin{proof}
	Let $g_\mu$ be the polynomial of the distribution $\mu$. Define
	\[f(z_1,\ldots,z_n)=\log g_\mu\big(\frac{z_1^\alpha}{\mu_1^\alpha},\ldots ,\frac{z_n^\alpha}{\mu_n^\alpha}\big).\]
	Since $\mu$ is $\alpha$ fractionally log-concave, $\log g_\mu(z_1^\alpha,\ldots,z_n^\alpha)$ is a concave function. Also, scaling preserves concavity. Therefore, $f(z_1,\ldots,z_n)$ is concave. 

Now, let $X$ be a random variable that indicates a set chosen according to $\mu$, i.e., $\mathbb{P}(X=\mathds{1}_S)=\mu(S)$.	
Then, by Jensen inequality we have that $f(\mathbb{E}[X])\geq \mathbb{E}[f(X)]$.
	Note that,
	\[f(\mathbb{E}(X))=f(\mu_1,\ldots,\mu_n)=\log g_\mu\big(\frac{\mu_1^\alpha}{\mu_1^\alpha},\ldots \frac{\mu_n^\alpha}{\mu_n^\alpha}\big)=\log(g_\mu(1,\ldots,1))=0.\]
	Also, 
	\[f(\mathds{1}_S)=\log(\sum_{T\subseteq S}\mu(T)\prod_{i\in T}\frac{1}{\mu(i)^\alpha})\geq \log(\mu(S)\prod_{i\in S}\frac{1}{\mu(i)^{\alpha}})=\log\mu(S)+\alpha\sum_{i\in S}\log\frac{1}{\mu(i)},\]
	where the inequality is true because of monotonicity of $\log$. Hence,
	\[\mathds{E}[f(X)]=\sum_S\mu(S)f(\mathds{1}_S)\geq \sum_S \left (\mu(S)\log\mu(S)+\alpha\mu(S)\sum_{i\in S}\log\frac{1}{\mu(i)} \right)\]
	\[=-\mathcal{H}(\mu)+\alpha\sum\mu(i)\log\frac{1}{\mu(i)}.\]
 By applying Jensen inequality we have
 \[\mathcal{H}(\mu)\geq \alpha\sum_i\mu(i)\log\frac{1}{\mu(i)}.\]
	\end{proof}

	Given a probability distribution $\mu:2^{[n]}\rightarrow\mathbb{R}_{\geq 0}$, the dual probability distribution  $\mu^*$ is defined so that the probability of occurrence of each set is equal to its complement under $\mu$, i.e., for any set $S\subseteq[n]$, $\mu^*(S)=\mu([n]\setminus S)$.
	\begin{corollary}
	If $\mu$ and its dual $\mu^*$ are $\alpha$-fractionally log-concave  then $\sum_i\mathcal{H}(\mu(i))$ is a $\frac{\alpha}{2}$ approximation for $\mathcal{H}(\mu)$.
	
	In particular, if $\mu$ is $\Gamma_{2\alpha}$-sector-stable, then $\mu$ and its dual $\mu^*$ are $\alpha$-fractionally log-concave (see \cref{prop: ssProperties} Part\ref{part:dual} and \cref{lemma:sectorStableToFracLC}). Therefore, $\sum_i\mathcal{H}(\mu(i))$ is a $\frac{\alpha}{2}$ approximation for $\mathcal{H}(\mu)$.
	\end{corollary}

	\begin{proof}
	For any probability distribution $\mu$ we have that, $\mathcal{H}(\mu)\leq \sum_i\mathcal{H}(\mu(i))$. So, it is enough to prove $\mathcal{H}(\mu)\geq\frac{\alpha}{2}\sum_i\mathcal{H}(\mu(i))$
	By  \cref{lemma: entropy lower bound} we have that,
	
	\[\mathcal{H}(\mu)\geq\alpha \sum_i\mu(i)\log\frac{1}{\mu(i)},\]
		\[\mathcal{H}(\mu^*)\geq\alpha \sum_i(1-\mu(i))\log\frac{1}{1-\mu(i)}.\]
		Since $\mu$ and $\mu^*$ are duals $\mathcal{H}(\mu)=\mathcal{H}(\mu^*)$. Therefore,
		\[2\mathcal{H}(\mu)=\mathcal{H}(\mu)+\mathcal{H}(\mu^*)\geq \alpha( \sum_i\mu(i)\log\frac{1}{\mu(i)}+ \sum_i(1-\mu(i))\log\frac{1}{1-\mu(i)})=\alpha\sum_i\mathcal{H}(\mu(i)). \]
		
	\end{proof}

Given a distribution $\mu$, let $F_\mu$ be its support. We want to show how to approximate $\log|F_\mu|$ when $\mu$ is fractionally log-concave. Previously, this result was shown for log-concave polynomials in \cite{AOV18}.

\begin{lemma} \label{lem:approximateLogSupport}
Consider $F\subseteq \binom{[n]}{k}.$ Let $F^*: = \set*{[n] \setminus S \given S \in F}.$ Let 
\[\beta := \max \set*{\sum_i p_i\log\frac{1}{p_i}\given p\in \conv(F)},\] and \[\beta^*:= \max\set*{q_i \log \frac{1}{q_i} \given q \in \conv(F^*)}.\] 
Assume there exists an $\alpha$-fractionally log-concave polynomials $g$ and $h$ with $\supp(g)=F$ and $\supp(h) = F^*.$ Then $\beta+\beta^*$ is $\alpha/2$-approximation for $\log \abs{F}$ i.e.
$(\beta+\beta^*) \geq \log \abs{F} \geq  (\beta + \beta^*) \alpha/2.$

In particular, if there exists an $\Gamma_{2\alpha}$-sector-stable polynomial $g$ with $\supp(g) = F$ then $\beta+\beta^*$ is $\alpha/2$-approximation for $\log \abs{F}.$
\end{lemma}
Note that $\beta$ and $\beta^*$ can be efficiently computed via a convex program (see e.g. \cite[Theorem 2.10]{AOV18}. The following lemma states that any point in $\conv(F_\mu)$, where $\mu$ is fractionally log-concave, is  equal to the marginals of some fractionally log-concave distribution.

\begin{lemma} \label{lem:two max equal}%two max are equal.
Consider $F \subseteq \binom{[n]}{k}.$ Suppose there exists an $\alpha$-fractionally log-concave polynomial $g$ with $\supp(g)=F$. For any $(p_1, \cdots, p_n) \in \conv(F),$ there exists $\nu$ with $\supp(\nu) \subseteq F$ such that $\sum_S \nu(S) z^S$ is $\alpha$-fractionally log concave and $\nu(i) = p_i \forall i\in [n].$

Consequently, $\max \set*{\sum_i p_i\log\frac{1}{p_i}\given p\in \conv(F)} = \max\set*{\sum_i \mu(i) \log\frac{1}{\mu(i)} \given \mu \in V }$ where $V$ is the set of $\alpha$-fractionally log concave $\mu$ with $\supp(\mu)\subseteq F.$
\end{lemma}

The proof is very similar to
\cite[Theorem 2.10, Corollary 2.11]{AOV18}. They show that given any $\vec{p}=(p_1, \cdots, p_n) \in \newt(g)$, one can find vector $\lambda \in \R^{n}_{\geq 0}$ such that the distribution generated by $g_{\ef{\lambda}{\mu}}$ has the same marginal as $\vec{p}.$ %Now, observe that if $g_{\mu}$ is $\alpha$-fractionally log-concave then so is $g_{\ef{\bmlambda}{\mu}}.$ 
 Now, we are ready to prove \cref{lem:approximateLogSupport}.
\begin{proof}[Proof of \cref{lem:approximateLogSupport}]
Let $\nu$ and $\nu^*$ be the uniform distribution over $F$ and $F^*$ respectively. For a set $F'$, let $V_{F'}$ be set of $\alpha$-fractionally log concave $\mu$ with $\supp(\mu)\subseteq F'$.
Since $V_{F}$ is non empty, by \cref{lem:two max equal}, $\beta = \max_{\mu \in V_F }\set*{\sum_i \mu(i) \log\frac{1}{\mu(i)} }$ and $\beta^* = \max_{\mu \in V_{F*} }\set*{\sum_i \mu(i) \log\frac{1}{\mu(i)} } .$ For $S \in \set*{F, F^*},$ let \[\mu^{\text{argmax}}_S = \text{argmax}_{\mu \in V_S}\sum_i \mu(i) \log\frac{1}{\mu(i)}. \]

We have $ \log (\abs{F}) = \mathcal{H}(\nu) \leq \sum_i \mathcal{H} (\nu(i)) = \sum_i (\nu(i) \log \frac{1}{\nu(i)} +(1-\nu(i)) \log \frac{1}{1-\nu(i)}) \leq \beta + \beta^*, $
 where the inequality follows from the fact that $(\nu(i))_{i=1}^n \in \conv(F)$ and $(1-\nu(i))_{i=1}^n \in \conv(F^*).$

On the other hand, since the uniform distribution over discrete set maximizes entropy, we have
\[\log (\abs{F}) = \mathcal{H}(\nu)  \geq \mathcal{H} (\mu^{\text{argmax}}_F) \geq \alpha \sum_i \mu^{\text{argmax}}_F(i) \log \frac{1}{\mu^{\text{argmax}}_F(i)}  = \alpha \beta,\]
where the second inequality follows from \cref{lemma: entropy lower bound}. Analogously,
\[\log (\abs{F^*}) = \mathcal{H}(\nu)  \geq \mathcal{H} (\mu^{\text{argmax}}_{F*}) \geq \alpha \sum_i \mu^{\text{argmax}}_{F*}(i) \log \frac{1}{\mu^{\text{argmax}}_{F*}(i)} = \alpha \beta^*. \]
Summing the these two inequalities, we get
\[\log (\abs{F}) \geq (\beta + \beta^*)\alpha/2.\]

\end{proof}

\cref{cor:homogenize monomer-dimer poly,cor:monomerDimerConstrained,lem:approximateLogSupport} together imply the following corollary. %an approximation for the number of size-$2k$-subsets of vertices that have a perfect matching. 
\begin{corollary}\label{cor:2kmatching}
Consider graph $G= G(V,E).$ Let $V^M$ be the family of sets of $S \subseteq V$ that have a perfect matching. For $k\leq n/2$, let $V^M_{k}$ be the family of vertices of size $2k$ that have a perfect matching. Then we can efficiently compute a $1/8$-multiplicative-approximation of $\log \abs{V^M}$ and of $\log \abs{V^M_k}.$
\end{corollary}
Analogously, \cref{cor:p0Constrained,lem:approximateLogSupport} together imply the following
\begin{corollary}
Consider matrix $L\in \R^{n\times n}$ such that $L+L^T$ is positive semi-definite. Let $V^L$ be the family of sets $S \subseteq [n]$ such that $\det(L_{S,S}) \neq 0$. For $k\leq n$, let $V^L_k$
be the family of sets $S \in \binom{[n]}{k}$ such that $\det(L_{S,S}) \neq 0.$ Then we can efficiently compute an $1/8$-multiplicative-approximation of $\log \abs{V^L}$ and of $\log \abs{V^L_k}.$
\end{corollary}
\begin{remark}
In \cref{lem:edgeBound}, we show the convex hull of the support of a $\Gamma_{\alpha}$-sector stable polynomial has edge length bounded by $O(1/\alpha).$ We can show a similar result for $\alpha$-fractionally log-concave polynomial. We leave the problem of characterizing the support of $\alpha$-fractionally log-concave polynomial to future work.
\end{remark}

	\section{Cardinality-Constrained Distributions}
In this section, we state and prove the formal version of \cref{cor:mixDerivative}. This result suggests that there is an efficient algorithm to compute mixed derivatives of a real-stable polynomials. The time complexity of the algorithm depends on the bit complexity of coefficients of the polynomial and the number of partial derivatives. As a result, we have an FPRAS to compute the sum of coefficients of the monomials corresponding to a partition matroid with constantly many parts. By dropping the assumption on the coefficients, the best known result gives an $e^r$-approximation factor where $r$ is the rank of matroid (see \cite{SV17}).
		\begin{lemma} \label{lem:mixDerivative}
		Let $f \in \R_{\geq 0}[z_1, \cdots, z_n]$ be a homogeneous real-stable polynomial whose maximum degree in $z_i$ is $\kappa_i$. Let $\kappa: =\sum_{i=1}^n \kappa_i.$ For $v^1, \cdots, v^k, x \in \R_{\geq 0}^n$ with $k=O(1)$, we can compute $\partial_{v^1}^{c_1} \cdots \partial_{v^k}^{c_k} f \mid_{\vec{z} = x}$ in polynomial time in $\kappa$ and $b$, where $b\geq 0$ is the bit complexity of the coefficients of $f$ and the entries of $v^1, \cdots, v^k, x$ i.e., these entries are in between $[-2^b, 2^b].$
		
		%We can compute $D_{v_1} D_{v_2} \cdots D_{v_k}$ in polynomial time in $d,$ as long as there is only constantly many different directional derivative $D_{v_i}$ with $v_i\in \R_{\geq 0}^n$,

% 		\june{we can polarize (converge homogeneous to multiaffine homogeneous) while preserve real stability}
	\end{lemma}
	\begin{proof}%[Proof of \cref{lem: mixDerivative}]
% 	\kiran{This is perfectly written, I will take another pass after writing related work.}
	W.l.o.g., we can assume $f$ is homogeneous multiaffine, else we replace $f$ with its polarization $f^{\uparrow} (z_{ij} )_{i\in [n], j \in [\kappa_i]}$ (see \cref{prop:polarization}). Note that $f^{\uparrow}$ is a homogeneous multi-affine polynomial in $\kappa$ variables, and has same degree as $f$. Moreover, $\partial_v f = (\sum_{i=1}^n \sum_{j=1}^{\kappa_i} v_i\frac{\partial}{\partial z_{i j}}) f^{\uparrow} .$ Each call to the oracle $\mathcal{O}_{f^{\uparrow}}$ for $f^{\uparrow}$ can be implemented using one call to the oracle $\mathcal{O}_{f}$ for $f.$  The bounds on coefficients of $f$ implies that the coefficients of $f^{\uparrow}$ are bounded by $2^{-\kappa^{O(1)}b}$ and $2^{\kappa^{O(1)}b}.$
% 	\kiran{Talk about computing derivative of $g$ in terms of derivative of $g^{\uparrow}$}.
	Therefore, in the remainder of the proof we assume that polynomial $f$ is multiaffine, homogeneous, real-stable, and $\kappa = n.$
	%Without loss of generality we can assume that the vectors $v_{i}$'s are not parallel to each other because %\kiran{write the reason here}.
	
	We divide the proof into two main steps. In the first step, we map the polynomial $f$ to another polynomial $g$ such that:
	\begin{enumerate}
	    \item $g$ is a homogeneous multiaffine real-stable polynomial in $n' = O(n)$ variables, of degree  $d = \deg (g) = \deg(f)$.
	    \item $D_{v^1}^{c_1} \cdots D_{v^k}^{c_k} (f)\mid_{x_1, \cdots, x_n}  = D_{w^1}^{c_1} \cdots D_{w^{k}}^{c_{k}} (g) \mid_{x'_1, \cdots, x'_{n'}}  \in \R$, where $x'_i \geq 0 ~\forall i\in [n']$. The vectors $w^i \in \set{0,1}^n$ and correspond to subsets $T_i \subseteq [n']$; further these sets $T_{i}$ are disjoint. Note that $D_{w^1}^{c_1} \cdots D_{w^{k}}^{c_{k}} (g)$ is exactly $h(x'_1, \cdots, x'_n)$ where $h(z_1, \cdots, z_{n'})$ is the sum of terms in $g$ whose $(T_1, \cdots, T_{k}, T_{k+1})$-degree is $(c_1, \cdots, c_{k}, c_{k+1})$ where $T_{k+1} = [n'] \setminus \bigcup_{i=1}^k T_i$ and $c_{k+1} = n' - \sum_{i=1}^k c_i$. 
	    %\kiran{should these be $\alpha'_{i}$}
	\end{enumerate}
	%In the second step, we compute the derivative of $g$, that is $D_{w^1}^{\alpha_i} \cdots D_{w^{k'}}^{\alpha_{k'}} (g)|_{x_1, \cdots, x_n}$ in polynomial time.
	In the second step, we (approximately) sample from the distribution $\mu $ generated by $h(x'_1 z_1, \cdots, x'_n z_n).$
	A routine sampling to counting argument then allows computing an approximation of $h(x'_1, \cdots, x'_n).$ 
	\cref{thm:partition} and \cref{lemma:sectorStableToFracLC} implies $h$ is $\alpha$-fractionally log-concave for $\alpha = 1/2^{k+2}.$ Let $\Delta = 4(1/\alpha -1),$  $\ell = \lceil\Delta \rceil,$ the $\ell$-steps down-up walk has eigenvalue gap $\geq \frac{1}{n^{\Delta +1} }.$ The bound on coefficients of $f$ implies an upper bound of $\kappa^{O(1)} b$ on $
	\min_{S \in \supp(\mu)}\log (1/\mu(S)),$ thus the random walk starting from any $S \in \supp(\mu)$ mixes in $k^{O(1)} b$ steps. We can use $\mathcal{O}_g$ to obtain a starting state in $\supp(\mu)$. Each step of the random walk can be implemented using polynomially many calls $\mathcal{O}_g.$ 
% 	For homogeneous polynomial, evaluation is the same as taking derivative $(\sum_{i=1}^n x_i \partial_i)^{d} \tilde{f} = \tilde{d}! f(x_1, \cdots, x_n)$ for homogeneous $\tilde{f}$ of degree $\tilde{d}.$ Thus $D_{v^1}^{c_1} \cdots D_{v^k}^{c_k} (f)\mid_{x_1, \cdots, x_n}  = \frac{1}{\tilde{d}!} D_x^{\tilde{d} } D_{v^1}^{c_1} \cdots D_{v^k}^{c_k} (f)$ with $\tilde{d} = d - \sum_{i=1}^k c_i.$

	For $t\leq n$ and $v_i> 0$ for $i\in [t]$, it is easy to see that,
	\[(\sum_{i=1}^t v_i \partial_i)^c f(x_1, \cdots, x_n) = (\sum_{i=1}^t \partial_i)^c  f(v_1 x_1, \cdots, v_t x_t, x_{t+1}, \cdots, x_n). \]
    For $j\in [k]$, consider $w^j\in \set*{0,1}^n$ where $w^j_i = \mathbb{1}(v^j_i\neq 0).$ For $i\in [n]$, let $x'_i := x_i \prod_{j: v^j_i \neq 0} v^j_i \geq 0.$ We have $D_{v^1}^{c_1} \cdots D_{v^k}^{c_k} (f)\mid_{x_1, \cdots, x_n}  = D_{w^1}^{c_1} \cdots D_{w^{k}}^{c_{k}} (f) \mid_{x'_1, \cdots, x'_{n'}} . $
	 
	Consider the linear transformation $T:\R[z_i] \to \R[z_{i,j}]_{i\in [n], j \in [k]}$ obtained by substituting $z_i: = \sum_{j=1}^k z_{i,j}$ for $i\in [n]$ and define $g = T(f)$. Clearly, $g$ is multiaffine, homogeneous and real-stable. %By Lemma \kiran{cite a reference here or some lemma}, it is immediate that $g$ is real-stable polynomial. 
	
	We next show that:
	\[T(D_{w^1}^{c_1} \cdots D_{w^k}^{c_k} f)|_{x_1, \cdots, x_n}  = D_{\tilde{w}^1}^{c_1} \cdots D_{\tilde{w}^{k}}^{c_{k}} (g)|_{\set*{\tilde{x}_{i,j}}_{i\in[n], j\in [k]}},\]
	where $\tilde{x}_{i,j} = x_i \forall j\in [k]$ and $\tilde{w}^j_{i,j'} = \begin{cases} w^j_i &\text{ if } j'=j\\ 0 &\text{ else} \end{cases}.$
	
	Recall that $f$ and $g$ are multiaffine.
	To prove the above equality, we only need to verify it for multiaffine monomials. %Let $r : = \sum_{j=1}^k c_j,$ and $m$ be some monomial. Recall that $f$ and $g$ are multi-affine. 
	Fix a multiaffine monomial $m$ and $j\in [k]$. We check
	\[T(w_{i_1}^{j}\cdots w_{i_{c_j}}^{j} \partial_{i_1}  \cdots \partial_{i_{c_j}} m)  = \tilde{w}_{i_1, j}^{j} \cdots \tilde{w}_{i_{c_j}, j}^{j}  (\prod_{t=1}^{c_j}
	\frac{\partial}{\partial z_{i_t, j}}) T(m). \]
	This immediately implies $T (D_{w^j}^{c_j} \tilde{f}) = D_{\tilde{w}^j}^{c_j} T(\tilde{f}) $ for any multiaffine polynomial $\tilde{f}.$ The desired equality then follows by induction.
	
	First, $w_{i_1}^{j}\cdots w_{i_{c_j}}^{j} = \tilde{w}_{i_1, j}^{j} \cdots \tilde{w}_{i_{c_j}, j}^{j} .$ We can factor them out, and prove $T( \partial_{i_1}  \cdots \partial_{i_{c_j}} m) = \prod_{t=1}^{c_j}
	\frac{\partial}{\partial z_{i_t, j}} T(m).$
	Now, if the $i_t$ are not distinct, then LHS and RHS are both $0$ since $m$ and $T(m)$ are multiaffine. If $m$ does not divide $z_{i_t}$ for some $t$, then both LHS and RHS are $0.$ Now, write $m = m_1\prod_{t=1}^{c_j} z_{i_t}$ for some monomials $m_1$ containing only variables in $[n] \setminus \set*{i_1, \cdots, i_{c_j}}.$ Clearly, the LHS is $T(m_1).$ The RHS is
	\[\prod_{t=1}^{c_j} \frac{\partial}{\partial x_{i_t, j}^{j}} T(m_1)\prod_{t=1}^{c_j} (\sum_{j=1}^k x_{i_t}^j ) = T(m_1) (\prod_{t=1}^ {c_j}\frac{\partial}{\partial x_{i_t,j}^{j}} )\left( \prod_{t=1}^{c_j} (\sum_{j=1}^k x_{i_t}^j )\right) = T(m_1). \]
	
	\end{proof}

	\PrintBibliography
\end{document}